\documentclass[12pt]{article}

\usepackage{amssymb} 
\usepackage{amsmath} 
\usepackage{amsthm}
\usepackage{enumerate}
\usepackage{array}
\usepackage{multirow}
\usepackage{float}
\usepackage{caption}
\usepackage{subcaption}
\usepackage{natbib}
\usepackage{graphicx}
\usepackage[a4paper,left=1in,right=1in,top=1in,bottom=1in]{geometry}
\PassOptionsToPackage{normalem}{ulem}
\usepackage{ulem}
\usepackage{xcolor}
\usepackage[
    colorlinks,
    linkcolor={blue!50!black},
    citecolor={blue!50!black},
    urlcolor={blue!80!black}]{hyperref}

\DeclareMathOperator*{\argmin}{arg\,min}
\theoremstyle{plain}
\newtheorem{thm}{\protect\theoremname}
\providecommand{\theoremname}{Theorem}

\begin{document}
\sloppy

\def\spacingset#1{\renewcommand{\baselinestretch}%
{#1}\small\normalsize} \spacingset{1}

\title{Estimation of Predictive Performance in High-Dimensional Data Settings using Learning Curves}

 \author{Jeroen M. Goedhart\footnote{Corresponding author, E-mail address: j.m.goedhart@amsterdamumc.nl} \textsuperscript{a}, Thomas Klausch\textsuperscript{a}, Mark A. van de Wiel\textsuperscript{a}\\
\textsuperscript{a}\small{Department of Epidemiology and Data Science, Amsterdam University Medical Centers}}
\maketitle

\begin{abstract}
In high-dimensional prediction settings, it remains challenging to
reliably estimate the test performance. To address this challenge,
a novel performance estimation framework is presented. This framework,
called Learn2Evaluate, is based on learning curves by fitting a smooth monotone
curve depicting test performance as a function of the sample size.
Learn2Evaluate has several advantages compared to commonly applied performance
estimation methodologies. Firstly, a learning curve offers a graphical overview of a learner. This overview
assists in assessing the potential benefit of adding training samples 
and it provides a more complete comparison between learners than 
performance estimates at a fixed subsample size. Secondly, a learning curve
facilitates in estimating the performance at the total sample size rather than a subsample size. Thirdly, Learn2Evaluate allows the computation 
of a theoretically justified and useful lower confidence bound. Furthermore, this bound
may be tightened by performing a bias correction. The benefits of
Learn2Evaluate are illustrated by a simulation study and applications to omics data.
\end{abstract}
\noindent%
\vfill

\newpage
\spacingset{1.5}

\section{Introduction}
\label{sec:intro}

Nowadays, many learners can predict a response from high-dimensional
data. To select the most accurate learner, it is key to reliably estimate
the predictive performance with a confidence bound. The
predictive performance may be quantified by many different metrics.
For classification, examples include the accuracy, the area under
the receiver operating curve (AUC) \citep{hanley1982seAUC}, and the
brier score \citep{brier1950verification}, whereas for regression
one usually reports the mean square error of prediction (PMSE). 

It is well-known that performance estimates are optimistically biased when they 
are based on the same data used to fit the learner.
This optimism bias is even more apparent for high-dimensional learners, as such learners 
overfit quickly. Performance estimates should therefore be based
on independent test data. However, for many high-dimensional
prediction studies, a test set is unavailable. Therefore,
learning the model and estimating the performance should be based
on a single, often small-sized, data set. In this paper, we present
a novel methodology to estimate the predictive performance in such
a scenario.

The standard method for performance estimation in small sample size
settings is resampling such as repeated hold-out
\citep{burman1989comparative}, K-fold cross-validation \citep{stone1974crossvalidation},
or bootstrapping \citep{efron1994introduction}. These techniques repeatedly
partition or resample the available data in a training set of a fixed
size, to fit the model, and a test set, to estimate the performance.
An acknowledged challenge for these techniques is to choose the size
of the training sets. Training sets that are close in size to the
complete data set lead to an almost unbiased predictive performance
estimate. However, such training sets also induce a large variability of the performance estimate
because the test sets are small. Resampling techniques are, therefore,
faced with a bias-variance trade-off, which leads to conservative
and often useless confidence bounds for the performance. Advances
in these techniques improve the uncertainty estimates \citep{michiels2005prediction,jiang2008calculating},
but neither provide a theoretical justification nor a point estimate.

Another performance estimation strategy is to combine resampling with an asymptotic
or parametric method to construct a confidence bound. For the cross-validated
AUC, a confidence bound based on influence functions was derived \citep{LeDell2015ComputationallyEC},
and Monte Carlo simulation was combined with a parametric model to
construct a confidence bound for the cross-validated accuracy \citep{dobbin2009method}.
However, the former bound is too liberal for small sample sizes and
the latter method only applies to linear learners. Moreover, both
methods lack generality because they only apply to a specific predictive
performance metric.

To address the aforementioned issues, we present Learn2Evaluate, a
learning-curve-based performance estimation framework. A learning
curve combines resampling with smoothing by fitting a monotone curve
depicting the predictive performance on left-out samples as a function
of the size of the training set \citep{cortes1994learning,mukherjee2003estimating}.
The learning curve enables us to obtain a point estimate of the performance at the full sample size 
and to determine an optimal training set size to construct a conservative confidence bound. Additionally, this bound can be tightened by performing a bias correction. 
The idea of finding an optimal training set size builds on previous work by \cite{dobbin2011optimally}.
They employed learning curves to find a good training set size in
single split sample approach. In contrast, we consider multiple splits
into a training and test set, and we focus on a confidence bound. 

Learn2Evaluate has several advantages compared to other performance estimation
techniques. Firstly, Learn2Evaluate is generic and applies
to any learner and any performance metric.
Secondly, it renders a point estimate and a confidence bound. Importantly,
this confidence bound is theoretically justified, as we will prove
type-I error control, and it explicitly takes the size of the training
set into account.

In addition, because Learn2Evaluate is based on learning curves, it
automatically offers a graphical technique to study several aspects
of a learner. Firstly, the slope of the curve at the total sample
size allows a qualitative assessment of the benefit of future samples.
Secondly, learning curves provide a more complete comparison between
different learners than just performance point estimates at a fixed
training set size. Lastly, one may observe that prediction models
learn at different rates, which could lead to a deeper understanding
of learning behaviors. 

The remainder of this paper is organized as follows. In \autoref{sec:meth},
we describe Learn2Evaluate and provide a theoretical justification for the confidence
bound. We present simulation results in \autoref{sec:Simulation} and in 
\autoref{sec:Applications}, we illustrate the benefits of Learn2Evaluate on omics data.
We end with some conclusive remarks in \autoref{sec:Discussion}.

\section{Methods}
\label{sec:meth}

\subsection{Target parameter}

Let our data $D=\left\{ \left(\boldsymbol{X}_{i},Y_{i}\right)\right\} _{i=1}^{N}$
consist of $N$ observations of a $p$-dimensional feature vector
$\boldsymbol{X}_{i}$ with a corresponding response variable $Y_{i}.$
The goal in prediction problems is to fit a learner $\varphi_{D}$
on data $D$ that reliably predicts $Y$ from
$\boldsymbol{X}.$ After learner $\varphi_{D}$ is fitted, we wish to
evaluate the predictive performance $\rho_{D},$ i.e. the quality
of predictions of $\varphi_{D}$ of yet unseen data points. Subscript
$D$ indicates that $\rho_{D}$ is conditional on the data $D.$ 

\subsection{Predictive performance estimation}

To evaluate $\rho_{D},$ we report a point estimate $\hat{\rho}_{D}$
and a lower confidence bound $L_{D}\left(\alpha\right)$
for $\rho_{D}$ that satisfies $P\left(L_{D}\left(\alpha\right)\leq\rho_{D}\right)\geq1-\alpha.$
We consider the case where a large test data set is not available. Learning $\varphi_{D}$ and estimating
$\hat{\rho}_{D}$ and $L_{D}\left(\alpha\right)$ should therefore be
solely based on data $D$. 

To use $D$ efficiently, one may apply resampling techniques. These
techniques create multiple training sets, on which $\varphi$ is learned,
and complementary test sets, on which the estimates $\hat{\rho}_{D}$ and $L_{D}\left(\alpha\right)$
are based. Since $\varphi$ is now learned on a subset of $D,$ it
is more difficult to find the structure in $\boldsymbol{X}$ that
explains $Y.$ This often leads to a pessimistically biased estimate
of $\rho_{D}.$ Larger training sets reduce this bias, but also lead to smaller test sets, 
which in turn increases the variance of the performance estimate. Resampling-based 
estimators are therefore faced with a bias-variance trade-off.
This trade-off makes it particularly difficult to obtain a practical
confidence bound as such a bound depends on both the bias and the variance. 

To address the bias-variance trade-off, we present a novel learning-curve-based performance estimation method. 
In the following, we formally define the learning curve and discuss
how we use it for predictive performance estimation. Here, we distinguish between increasing 
performance metrics such as the AUC, which are metrics that increase when more signal is captured, and decreasing performance metrics 
such as the PMSE. Our method is presented for the AUC because classification is a frequent application 
in high-dimensional and small sample size settings. Extending our method to
decreasing performance metrics is straightforward and is explained for the PMSE.
Our method, which we call Learn2Evaluate, is summarized in \autoref{fig:Learning-Curve-methodology}.

\subsection{Definition of the learning curve}

A learning curve is constructed by first applying repeated hold-out estimation for several training
set sizes, followed by a smoothing step.

\subsubsection{Repeated hold-out estimation}
To define repeated hold-out estimation, we first introduce single hold-out
estimation. Let $s_{n}^{k}$ be the $k\mbox{th}$ realization of a 
random subset of $D$ with size $n\leq N$. This subset
is taken without replacement and is balanced with respect
to the response if $Y$ is binary. For continuous $Y,$ we do not impose the balancing restriction. 
We then define $\rho_{s_{n}^{k}}$ as
the predictive performance when $s_{n}^{k}$ is used to build a learner
$\varphi_{s_{n}^{k}}$. We estimate $\rho_{s_{n}^{k}}$ based on the
complement or hold-out set $D\setminus s_{n}^{k}.$ This yields the
single hold-out performance estimate $\hat{\rho}_{s_{n}^{k}}$  
and a lower confidence bound $L_{s_{n}^{k}}(\alpha).$
We assume that a method to construct $L_{s_{n}^{k}}(\alpha)$ exists. 
For the AUC, we employ the nonparametric method of \cite{delong1988comparing} implemented
in the R package pROC \citep{pROC}. For decreasing performance metrics, such as the PMSE, we estimate an upper confidence bound.

The single hold-out estimate depends on the subset $s_{n}^{k}$.
To reduce this dependency, we apply repeated hold-out estimation. We
define $\rho_{n}$ as the average performance over all $V$ possible subsets 
of $D$ of size $n\leq N,$
\[
\rho_{n}=\frac{1}{V}\sum_{k=1}^{V}\rho_{s_{n}^{k}},
\]
with $V=\binom{N}{n}$ for unrestricted sampling, and 
$V=\binom{N_{+}}{n_{+}}\binom{N_{-}}{n_{-}}$ for balanced sampling, with
subscript $+/-$ denoting the number of positive/negative cases of the given
(sub)sample size. The average performance is estimated by the repeated 
hold-out estimate $\hat{\rho}_{n}$, which is based on $K\leq V$ subsets,
\begin{equation}
\hat{\rho}_{n}=\frac{1}{K}\sum_{k=1}^{K}\hat{\rho}_{s_{n}^{k}}.\label{eq:repeated hold-out estimate}
\end{equation}
Repeated hold-out estimation also yields $K$ lower confidence bounds $L_{s_{n}^{k}}.$
In \hyperref[thm:proof of conf bound]{Theorem 1}, we prove how to aggregate these bounds to obtain a $(1-\alpha)$ confidence bound for $\rho_{D}.$

Conventional repeated hold-out estimation uses \eqref{eq:repeated hold-out estimate}
as a point estimate for the performance $\rho_{D}.$ A
learning curve, however, relies on first evaluating \eqref{eq:repeated hold-out estimate}
for a sequence of $J$ subsample sizes $n_{1}<\ldots<n_{J}$, yielding
a learning trajectory (\autoref{fig:Learning-Curve-methodology},
black dots), 
\begin{equation}
\left\{(n_{j},\hat{\rho}_{n_{j}})\right\} _{j=1}^{J}.\label{eq:learning trajectory}
\end{equation}
The learning curve is then obtained by smoothing this trajectory and the fit 
provides an estimate of $\rho_{D}.$

\subsubsection{Smoothing}
A learning curve $f(n)$ (\autoref{fig:Learning-Curve-methodology},
solid purple line) is a fit of the learning trajectory, i.e. \eqref{eq:learning trajectory}, 
rendering a relationship between the subsample size and the predictive performance. 
We assume that this relationship is monotonically increasing, 
\begin{equation}
\forall j>0:\,\rho_{n_{j}}\leq\rho_{n_{j+1}}.\label{eq:monotony}
\end{equation}
Assumption \eqref{eq:monotony} is a natural one:
the average performance of a learner trained on larger subsets is
likely higher than that of the same learner trained on smaller subsets.

In addition, we assume that the learning curve is
concave after some threshold subsample size $n_{tr}$,  $\forall n_{j}\geq n_{tr},j>0:\,\left(\rho_{n_{j+2}}-\rho_{n_{j+1}}\right)/\left(n_{j+2}-n_{j+1}\right)\leq\left(\rho_{n_{j+1}}-\rho_{n_{j}}\right)/\left(n_{j+1}-n_{j}\right).$
This assumption is employed to stabilize the curve. 
In some situations, the concavity assumption may not hold for the considered trajectory of subsample sizes, i.e. $n_{1}\leq n_{tr}.$
This can be assessed visually from the learning trajectory.
For example, in some situations, lasso regression is known to exhibit a phase transition \citep{PhaseTransitionDonoho}, 
which would lead to a S-shaped learning curve. For such an instance, we propose to first determine $n_{tr}$, i.e. the inflection point of the S-shaped curve, followed by fitting
a learning curve with smallest subsample size $n_{1}=n_{tr}.$ In \hyperref[sec:Supplements]{supplementary Section 1}, we show a method to determine $n_{tr}.$ 
Note that the case $n_{J}\leq n_{tr}$ 
is a strong indication for acquiring more samples before developing a learner.

For decreasing performance metrics, the assumptions should be reversed, leading to a
monotonically decreasing and convex curve.

To fit the learning curve, we consider a parametric model, based on
an inverse power law, and a nonparametric model, based on a constrained regression
spline. Use of power laws for fitting learning curves has been empirically
established \citep{cortes1994learning,mukherjee2003estimating,hess2010learning}.
We examine constrained regression splines because they do not rely
on a specific shape of $f,$ while the constraints prevent a too flexible
fit.

\paragraph{Inverse power law}\mbox{}\\
To fit an inverse power law, we perform nonlinear least squares regression, with
\begin{equation}
f(n)=\delta-\beta n^{-\gamma}.\label{eq:inverse power law}
\end{equation}
This function grows, with increasing subsample size $n$, asymptotically
to parameter $\delta$, with $0.5\leq\delta\leq1$ for the AUC. 
The learning rate $\beta$ and the decay rate $\gamma$ are positive \citep{cortes1994learning}.
Derivative inspection reveals that \eqref{eq:inverse power law}
satisfies the monotony and concavity assumptions. To estimate the
parameters $\delta,$ $\beta$, and $\gamma,$ we solve
\begin{equation}
\left(\hat{\delta},\hat{\beta},\hat{\gamma}\right)=\argmin_{\delta,\beta,\gamma}\frac{1}{J}\sum_{j=1}^{J}\left(\delta-\beta n_{j}^{-\gamma}-\hat{\rho}_{n_{j}}\right)^{2},\quad\mbox{s.t.}\enskip0.5\leq\delta\leq1,\enskip0\leq\beta,\enskip\mbox{and}\enskip0\leq\gamma, \label{eq:fitting the power law}
\end{equation}
by using the L-BFGS-B method \citep{byrd1995limited} of the base R optim function.
For decreasing metrics, we have $f(n)=\delta+\beta n^{-\gamma},$ with 
$0\leq\delta,\enskip0\leq\beta,\enskip\mbox{and}\enskip0\leq\gamma.$

\paragraph{Constrained regression spline}\mbox{}\\
To fit a constrained regression spline,
we employ an algorithm proposed by \cite{Ng2007} 
implemented in the R package COBS \citep{Ng2020}.
Briefly, the function $f$ is approximated by fitting piecewise polynomial
basis functions over different regions of the subsample size $n$.
These regions are specified by $M$ knots, which leads to a B-spline
basis $\boldsymbol{\pi}\left(n\right)\in R^{M+\kappa},$ $\boldsymbol{\pi}\left(n\right)=\left(\pi_{1}\left(n\right),\ldots,\:\pi_{T+\kappa}\left(n\right)\right)^{T},$
where $\kappa$ is the maximum degree of the polynomial. Then, we
define the learning curve as 
\begin{equation}
f(n)=\boldsymbol{\pi}\left(n\right)^{T}\boldsymbol{\theta},\label{eq:spline definition}
\end{equation}
 where $\boldsymbol{\theta}\in R^{M+\kappa},$ a vector of regression
parameters, is estimated by solving
\begin{equation}
\boldsymbol{\hat{\theta}}=\argmin_{\boldsymbol{\theta}}\sum_{j=1}^{J}\mid\hat{\rho}_{n_{j}}-\boldsymbol{\pi}\left(n_{j}\right)^{T}\boldsymbol{\theta}\mid,\quad\mbox{subject to}\label{eq: finding optimal spline}
\end{equation}
\[
\boldsymbol{\pi}(n_{j})^{T}\boldsymbol{\theta}\leq\boldsymbol{\pi}(n_{j+1})^{T}\boldsymbol{\theta}\quad\mbox{and}\quad\frac{\boldsymbol{\pi}(n_{j+1})^{T}\boldsymbol{\theta}-\pi(n_{j})^{T}\boldsymbol{\theta}}{\left(n_{j+1}-n_{j}\right)}\leq\frac{\boldsymbol{\pi}(n_{j})^{T}\boldsymbol{\theta}-\boldsymbol{\pi}(n_{j-1})^{T}\boldsymbol{\theta}}{\left(n_{j}-n_{j-1}\right)}.
\]
Again, the constraints should be decreasing and convex for decreasing performance metrics.

\subsection{Learning-curve-based performance estimates}

Once the learning curve is fitted, we use it to obtain a point estimate
and a confidence bound of the target performance $\rho_{D}.$ 

\subsubsection{Point estimate}
For the point estimate, we extrapolate the learning curve to the full sample size
$N,$ i.e. we estimate $\rho_{D}$ by $f(N)$ (\autoref{fig:Learning-Curve-methodology},
black cross). Extrapolation to $N$ circumvents the pessimistic bias
of resampling techniques, induced by learning the model on a subset
of $D.$ Additionally, smoothing of the learning trajectory, combined with the constraints, 
may denoise the individual repeated hold-out estimates. 

\subsubsection{Lower confidence bound}
For the confidence bound, we employ the learning curve to find a good
training set size that takes the bias-variance trade-off into account. 

We first prove a $(1-\alpha)$ lower confidence bound for parameter $\rho_{D}.$
We start by posing the following assumption: 
\begin{equation}\label{meanprob}
\frac{1}{V}\sum_{k=1}^V P(L_{s_n^k} > \rho_D) \leq \frac{1}{V}\sum_{k=1}^V P(L_{s_n^k} > \rho_{s_n^k})
\end{equation}
Note that this assumption aligns well with the monotony assumption  \eqref{eq:monotony}: for $\rho_D$ being larger than $\rho_{s_n^k}$ on average it is reasonable to assume that the average exceedance probability for threshold $\rho_D$ is smaller than that for $\rho_{s_n^k}$. 
\begin{thm}
\label{thm:proof of conf bound} Let $L_{s_{n}^{k}}(\alpha)$ be the $(1-\alpha)$ lower confidence bound for $\rho_{s_{n}^{k}},$ and let $S_n^k$ be a {\it random} subset of size $n$, with corresponding random lower bound $L_{S_{n}^{k}}(\alpha)$.  
Assume \eqref{meanprob} and suppose there exists a monotone transformation
$\xi$ of $L_{S_{n}^{k}}$ such that $(\xi(L_{S_{n}^{k}}))_{k=1}^K$
follows a multivariate normal distribution with standard normal marginals, then
\[
\forall n\leq N:\,P\left(L_{n}(\alpha)\leq\rho_{D}\right)\geq 1-\alpha,
\]
where $L_{n}(\alpha) = \text{med}\bigl(L_{S_{n}^{k}}(\alpha)\bigr)_{k=1}^K$.
\end{thm}
\begin{proof}
First note that from \eqref{meanprob} we have
\begin{equation}\label{Lsnk}
\begin{split}
P(L_{S_n^k} \leq \rho_D) &= \sum_{k=1}^V P(S_n^k=s_n^k) P(L_{S_n^k} \leq \rho_D | S_n^k=s_n^k)
\\ &= \frac{1}{V}\sum_{k=1}^V P(L_{s_n^k} \leq \rho_D) \geq \frac{1}{V}\sum_{k=1}^V P(L_{s_n^k} \leq \rho_{s_n^k}) = 1-\alpha,
\end{split}
\end{equation}
as $L_{s_n^k}$ is a $(1-\alpha)$ lower confidence bound for $\rho_{s_{n}^{k}}$. Next, we have
\begin{equation}\label{Ln}
\begin{split}
P\left(L_{n}(\alpha)\leq \rho_D \right) &= P\biggl(\text{med}\bigl(L_{S_{n}^{k}}(\alpha)\bigr)_{k=1}^K   \leq \rho_D \biggr)\\ &=
P\biggl(\xi\Bigl(\text{med}\bigl(L_{S_{n}^{k}}(\alpha)\bigr)_{k=1}^K\Bigr)  \leq \xi(\rho_D)  \biggr)\\ &=
P\biggl(\text{med}\Bigl(\xi\bigl(L_{S_{n}^{k}}(\alpha)\bigr)\Bigr)_{k=1}^K  \leq  \xi(\rho_D)  \biggr).
\end{split}
\end{equation}
Then, \eqref{Lsnk} and $\xi$ being monotone imply $P\bigl(\xi(L_{S_{n}^{k}}(\alpha)) \leq \xi(\rho_D)\bigr) \geq 1-\alpha$. Hence $\xi(\rho_D) \geq \Phi^{-1}(1-\alpha),$ with $\Phi$ the univariate standard normal c.d.f., as $\xi\left(L_{S_{n}^{k}}(\alpha)\right) \sim N(0,1).$ Next, we follow Th.1 in \citep{van2009testing}, which states that for multivariate normal and marginally standard normal random variables $Z_k$: $P\left(\text{med}(Z_k)_{k=1}^K \leq x\right) \geq \Phi(x)$.
Finally, substitute $x= \xi(\rho_D)$:
$$P\biggl(\text{med}\Bigl(\xi\bigl(L_{S_{n}^{k}}(\alpha)\bigr)\Bigr)_{k=1}^K  \leq  \xi(\rho_D)\biggr) \geq \Phi\bigl(\xi(\rho_D)\bigr) \geq \Phi\left(\Phi^{-1}(1-\alpha)\right) = 1-\alpha,$$
which, in combination with \eqref{Ln}, completes the proof.
\end{proof}
Note that the strengths of \hyperref[thm:proof of conf bound]{Theorem 1} are that 1) it only requires the existence of $\xi$; we do not need to know what it actually is;
and 2) for any continuous random variable $Z$ with finite mean there always exist a monotone $\xi$ such that {\it marginally} $\xi(Z) \sim N(0,1)$. Because a multivariate normal distribution can always be decomposed as the sum of the marginals and the copula \citep{sklar1959fonctions}, the multivariate assumption in \hyperref[thm:proof of conf bound]{Theorem 1} is effectively only an assumption on the dependencies between copies of $\xi(L_{S_n^k})$. 

\hyperref[thm:proof of conf bound]{Theorem 1} also holds for decreasing performance metrics and an upper confidence bound.

\paragraph{Determining a good subsample size}\mbox{}\\
\hyperref[thm:proof of conf bound]{Theorem 1} implies that we can
determine a lower confidence bound $L_{n}(\alpha)$ for
$\rho_{D}$ at any subsample size $n\leq N.$ To determine an optimal $n,$ we explicitly take the bias-variance trade-off into account 
by employing the saturation level of the learning curve. 
When a learner has converged at subsample size $n \ll N,$ relatively many samples can
be used for testing because this will hardly impact the (empirical) bias, i.e. $\widehat{\mbox{bias}}(n)=f(N)-f(n).$
The learning curve enables us to find a $n,$ which we call $n_{opt},$ such that 
the bias is not too large, and the test set is not too small, which would lead to high variance.
To determine $n_{opt}$, we minimize the mean square error, i.e. 
\begin{equation}
n_{opt}=\argmin_{n} \left\{\widehat{\mbox{bias}}^{2}(n)+\widehat{\mbox{Var}}(n)\right\},\label{eq: nopt mse}
\end{equation}
with $\widehat{\mbox{bias}}(n)$ defined above. For $\widehat{\mbox{Var}}(n)$, we 
employ an asymptotic variance estimator for the given performance metric. For the AUC, 
we employ a formula derived by \cite{bamber1975area}, and for the PMSE, \cite{FABER199979}
derived an estimator. These estimators are found in \hyperref[sec:Supplements]{supplementary Section 2}.

If no asymptotic variance estimator is available, we suggest to control for the empirical bias. This method is described in \hyperref[sec:Supplements]{supplementary Section 3}. 
\begin{figure}[H]
\begin{center}
\includegraphics[bb=13bp 20bp 960bp 540bp,clip,width=6.5in]{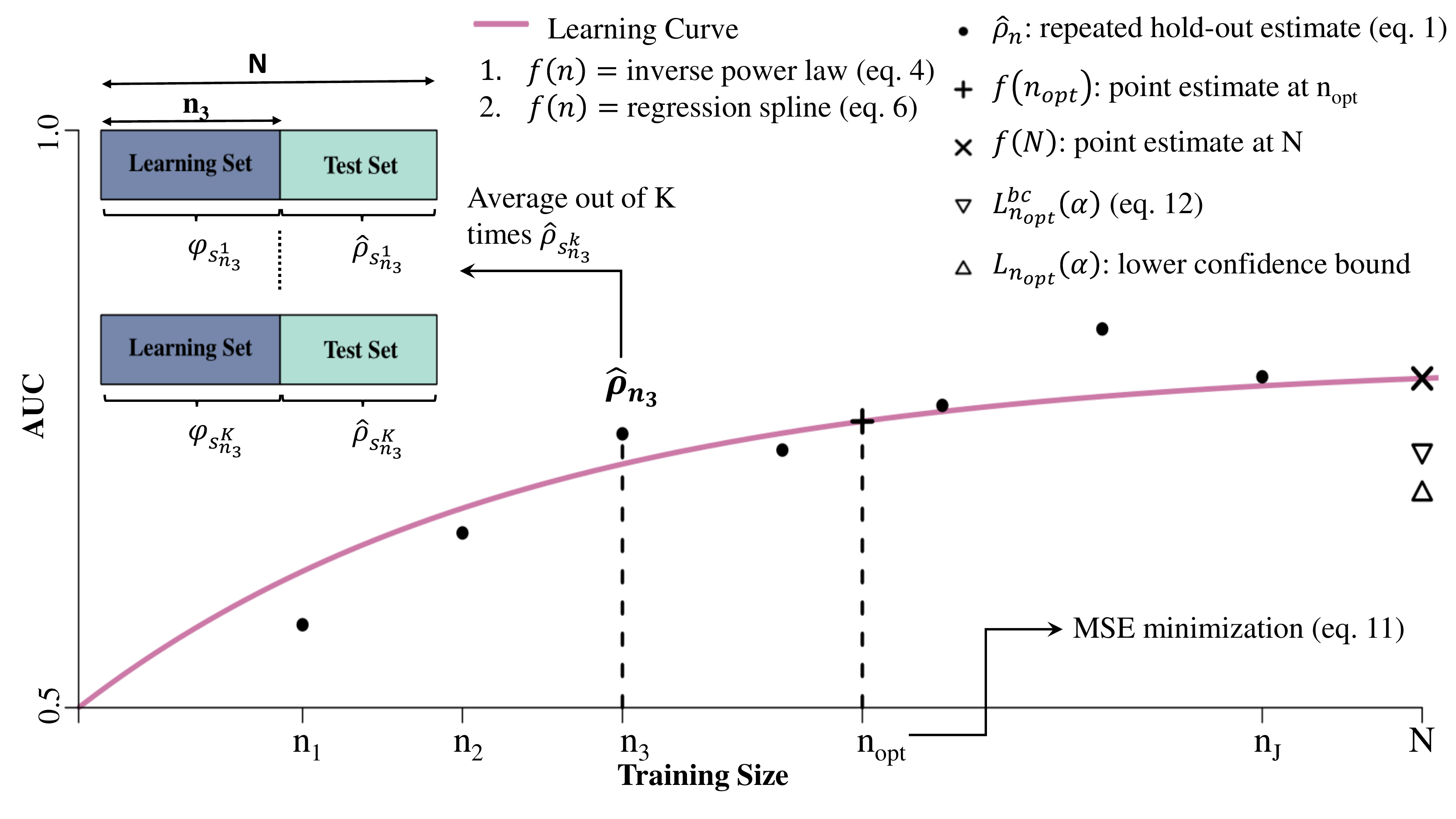}
\end{center}
\caption{\footnotesize{Summary of Learn2Evaluate.  First, we estimate the $AUC$ for
several training sample sizes $n_{j}$ by repeated hold-out estimation, rendering a learning
trajectory (black dots). We smooth this trajectory by an inverse
power law or a constrained regression spline, leading to a 
relationship $f(n)$ between the training set size $n$ and the AUC.
We then estimate the predictive performance $\rho_{D}$ by $f(N)$ (black cross). 
The learning curve also allows us to determine a lower confidence bound (upward
triangle) at an optimal training set size $n_{opt}$ (dotted vertical
line), which we determine by MSE minimization.
To obtain a tighter confidence bound, a bias correction is performed (downward triangle).}}
\label{fig:Learning-Curve-methodology}
\end{figure}
\subsubsection{Bias-corrected confidence bound}
Our predictive performance point estimate for $\rho_{D}$ is
$f(N)$ (\autoref{fig:Learning-Curve-methodology},
black cross) and our confidence bound for $\rho_{D}$ equals $L_{n_{opt}}(\alpha)$
(\autoref{fig:Learning-Curve-methodology}, upward triangle), with
$n_{opt}$ determined by \eqref{eq: nopt mse}.
This confidence bound may be tightened by performing a bias correction.
$L_{n_{opt}}$ is in principle a confidence bound for $\rho_{n_{opt}},$ although it is also valid for $\rho_{D}.$ Hence, we may
perform an empirical bias correction to $L_{n_{opt}}$ 
(\autoref{fig:Learning-Curve-methodology}, downward triangle), which leads to
\begin{equation}
L^{bc}_{n_{opt}}(\alpha)=L_{n_{opt}}(\alpha)+\left(f(N)-f(n_{opt})\right).\label{eq: bias corrected conf bound}
\end{equation}
Correcting for the empirical bias depends on the validity of our point estimate $f(N),$ which we assess in a simulation study.
\section{Simulation study}
\label{sec:Simulation}

To verify Learn2Evaluate, we evaluate the quality of the point estimate
$f(N),$ and the theoretically justified confidence bound $L_{n_{opt}}(\alpha)$ and
its bias-corrected version $L^{bc}_{n_{opt}}(\alpha)$ in a high-dimensional simulation study. We 
focus on classification and the AUC as performance metric.

\subsection{Description of experiments}

The binary response  is generated via
\begin{equation}
Y_{i}\sim \mbox{Bern}\left(\frac{\exp\left(\boldsymbol{X}_{i}\boldsymbol{\beta}\right)}{1+\exp\left(\boldsymbol{X}_{i}\boldsymbol{\beta}\right)}\right) \qquad i=1,\ldots,N, \label{eq: Simulated Response}
\end{equation}
with covariates $\boldsymbol{X}_{i}\sim \mathcal{N}\left(0_{p},\varSigma_{p\times p}\right),$ where $p=2000$, and elements of the parameter vector $\beta_{1}\ldots\beta_{p}\sim \mbox{Expo}(\nu).$  
To mimic a realistic correlation structure between the features,
we define $\varSigma_{p\times p}$ as the estimated variance-covariance matrix  from a real high-dimensional
omics data set described in \cite{best2015rna}, from which we randomly
select $2000$ genes as covariates. This matrix is estimated
by a shrinkage method \citep{SchaferStrimmer2005} implemented in the R package corpcor \citep{schafer2017package}.
Parameter vector $\boldsymbol{\beta}$  is defined such that a dense covariate-response structure is present, which is often the case in 
omics applications \citep{BOYLE2017}. We vary the sample size $N$ ($N=100$ and
$N=200$) and the signal in the data by tuning rate parameter $\nu$ ($\nu=1000$ and
$\nu=100$). 

For each setting, specified by $N$ and $\nu$, we simulate $R=1000$
data sets. We then apply Learn2Evaluate to each data set for three learners:
ridge regression, lasso regression, both implemented in the R package  glmnet \citep{glmnet},  and random forest, for which we employ the R package randomforestSRC  \citep{randomforestSRC}.
We generate learning curves by first estimating the AUC by repeated hold-out estimation
($50$ repeats) for ten subsample sizes $n_{j}$ homogeneously spread
over the sample size range $\left[20,N-10\right]$. For ridge and
lasso regression, we estimate the penalty parameter for each $n_{j}$ 
by the median of repeated ($5$ times) $10$-fold CV. This median is fixed 
for the $50$ repeats of each $n_{j}$. We obtain the learning
curve by smoothing the AUC point estimates at different subsample
sizes by either the inverse power law or the constrained regression spline.
We then estimate the AUC by $f(N)$ and 95\% lower confidence
bounds by $L_{n_{opt}}(0.05)$ and its bias-corrected version $L^{bc}_{n_{opt}}(0.05).$ We determine $n_{opt}$
by MSE minimization given by \eqref{eq: nopt mse}, using the asymptotic AUC variance estimator (eq. 5 of \hyperref[sec:Supplements]{supplementary Section 2}). 

We denote $f(N)$ by $f_{r}(N)$ for the $r$th simulated data set, and we approximate the true AUC by $\mbox{AUC}_{r}$: 
the AUC of the $r$th learned model evaluated on $25,000$ independent test samples.
The average $\mbox{AUC}_{r}$ (of $R=1000$ simulation runs) is given in the Appendix (\autoref{tab:True-AUC-results}). 

We assess the quality of point estimates $f_{r}(N)$ by
the root mean square error (RMSE) and the averaged bias (Bias),
\begin{equation}
\mbox{RMSE}=\sqrt{\frac{1}{R}\sum_{r=1}^{R}\left(f_{r}(N)-\mbox{AUC}_{r}\right)^{2}}\quad\mbox{and}\quad\mbox{Bias}=\frac{1}{R}\sum_{r=1}^{R}\left(f_{r}(N)-\mbox{AUC}_{r}\right).\label{eq: RMSE simulation}
\end{equation}
Confidence bounds $L_{n_{opt}}(0.05)$ and $L^{bc}_{n_{opt}}(0.05)$
are evaluated by their coverage, i.e. the proportion of times the confidence bound
is lower than $\mbox{AUC}_{r},$ and the distance of each confidence bound to its corresponding $\mbox{AUC}_{r}.$

We compare Learn2Evaluate, fitted by an inverse power law (L2E-P) or a spline (L2E-S), with two methods that produce a point estimate
and a lower confidence bound. Firstly, we consider ten-fold cross-validation (10F-CV). Point estimates are given by
the average of the ten test folds and a lower confidence bound is obtained by employing an asymptotic
variance estimator, which is derived by \cite{LeDell2015ComputationallyEC} and implemented 
in the R package cvAUC \citep{ledell2014package}.
Secondly, we consider leave-one-out bootstrapping (LOOB) \citep{efron1994introduction}
with $500$ bootstrapped training sets and complementary
test sets. From the $500$ AUC estimates, we obtain a point estimate 
and a lower confidence bound by taking the average and the 5\% quantile, respectively. 
Details on the code are found in \hyperref[sec:Supplements]{supplementary Section 9}.

\subsection{Point estimate evaluation}

The RMSE and the averaged bias of the point estimates are given in \autoref{tab:point estimates}.
\begin{table}[!t]
\centering{}\caption{\footnotesize{Root mean square error (RMSE)
and averaged bias (Bias) of Learn2Evaluate, fitted by an inverse power law (L2E-P) or a constrained
regression spline (L2E-S), 10-fold CV (10F-CV), and leave-one-out bootstrapping
(LOOB).}}
\begin{tabular}{cc>{\centering}p{1.1cm}>{\centering}p{1.1cm}>{\centering}p{1.1cm}>{\centering}p{1.1cm}>{\centering}p{1.1cm}>{\centering}p{1.1cm}}
\hline 
\multicolumn{2}{c}{} & \multicolumn{2}{c}{\textbf{\footnotesize{}Ridge}} & \multicolumn{2}{c}{\textbf{\footnotesize{}Lasso}} & \multicolumn{2}{c}{\textbf{\footnotesize{}RF}}\tabularnewline
\hline 
\multicolumn{2}{c}{} & {\footnotesize{}RMSE} & {\footnotesize{}Bias} & {\footnotesize{}RMSE} & {\footnotesize{}Bias} & {\footnotesize{}RMSE} & {\footnotesize{}Bias}\tabularnewline
\cline{3-8} \cline{4-8} \cline{5-8} \cline{6-8} \cline{7-8} \cline{8-8} 
 & {\footnotesize{}L2E-P} & \textbf{\footnotesize{}0.053} & {\footnotesize{}0.004} & {\footnotesize{}0.058} & \textbf{\footnotesize{}0.009} & \textbf{\footnotesize{}0.048} & {\footnotesize{}0.020}\tabularnewline
{\footnotesize{}$N=100,$} & {\footnotesize{}L2E-S} & {\footnotesize{}0.054} & {\footnotesize{}0.007} & {\footnotesize{}0.070} & {\footnotesize{}0.025} & {\footnotesize{}0.050} & {\footnotesize{}0.026}\tabularnewline
{\footnotesize{}$\nu=1000$} & {\footnotesize{}10F-CV} & {\footnotesize{}0.056} & \textbf{\footnotesize{}-0.001} & {\footnotesize{}0.084} & \textbf{\footnotesize{}-0.009} & {\footnotesize{}0.051} & \textbf{\footnotesize{}0.017}\tabularnewline
 & {\footnotesize{}LOOB} & {\footnotesize{}0.071} & {\footnotesize{}-0.039} & \textbf{\footnotesize{}0.057} & {\footnotesize{}-0.020} & {\footnotesize{}0.050} & {\footnotesize{}-0.013}\tabularnewline
\hline 
 & {\footnotesize{}L2E-P} & \textbf{\footnotesize{}0.037} & {\footnotesize{}0.002} & \textbf{\footnotesize{}0.048} & \textbf{\footnotesize{}0.002} & \textbf{\footnotesize{}0.037} & {\footnotesize{}0.002}\tabularnewline
{\footnotesize{}$N=100,$} & {\footnotesize{}L2E-S} & {\footnotesize{}0.038} & {\footnotesize{}0.004} & {\footnotesize{}0.056} & {\footnotesize{}0.009} & {\footnotesize{}0.038} & {\footnotesize{}0.005}\tabularnewline
{\footnotesize{}$\nu=100$} & {\footnotesize{}10F-CV} & {\footnotesize{}0.039} & \textbf{\footnotesize{}0} & {\footnotesize{}0.062} & {\footnotesize{}-0.005} & {\footnotesize{}0.039} & \textbf{\footnotesize{}0.001}\tabularnewline
 & {\footnotesize{}LOOB} & {\footnotesize{}0.048} & {\footnotesize{}-0.028} & {\footnotesize{}0.051} & {\footnotesize{}-0.025} & {\footnotesize{}0.043} & {\footnotesize{}-0.018}\tabularnewline
\hline 
 & {\footnotesize{}L2E-P} & \textbf{\footnotesize{}0.036} & \textbf{\footnotesize{}0.001} & \textbf{\footnotesize{}0.041} & \textbf{\footnotesize{}0.004} & \textbf{{\footnotesize{}0.034}} & \textbf{\footnotesize{}0.002}\tabularnewline
{\footnotesize{}$N=200,$} & {\footnotesize{}L2E-S} & {\footnotesize{}0.038} & {\footnotesize{}0.004} & {\footnotesize{}0.046} & {\footnotesize{}0.010} & \textbf{{\footnotesize{}0.034}} & {\footnotesize{}0.004}\tabularnewline
{\footnotesize{}$\nu=1000$} & {\footnotesize{}10F-CV} & {\footnotesize{}0.039} & {\footnotesize{}-0.002} & {\footnotesize{}0.050} & \textbf{\footnotesize{}-0.004} & {\footnotesize{}0.035} & {\footnotesize{}0.009}\tabularnewline
 & {\footnotesize{}LOOB} & {\footnotesize{}0.060} & {\footnotesize{}-0.041} & {\footnotesize{}0.052} & {\footnotesize{}-0.034} & {\footnotesize{}0.036} & {\footnotesize{}-0.009}\tabularnewline
\hline 
 & {\footnotesize{}L2E-P} & \textbf{\footnotesize{}0.024} & \textbf{\footnotesize{}0} & \textbf{\footnotesize{}0.031} & \textbf{\footnotesize{}0.001} & \textbf{{\footnotesize{}0.026}} & {\footnotesize{}0.007}\tabularnewline
{\footnotesize{}$N=200,$} & {\footnotesize{}L2E-S} & {\footnotesize{}0.025} & {\footnotesize{}0.002} & {\footnotesize{}0.035} & {\footnotesize{}0.004} & {\footnotesize{}0.027} & {\footnotesize{}0.008}\tabularnewline
{\footnotesize{}$\nu=100$} & {\footnotesize{}10F-CV} & {\footnotesize{}0.025} & {\footnotesize{}-0.001} & {\footnotesize{}0.036} & {\footnotesize{}-0.005} & \textbf{{\footnotesize{}0.026}} & {\footnotesize{}0.008}\tabularnewline
 & {\footnotesize{}LOOB} & {\footnotesize{}0.035} & {\footnotesize{}-0.022} & {\footnotesize{}0.037} & {\footnotesize{}-0.022} & \textbf{{\footnotesize{}0.026}} & \textbf{\footnotesize{}-0.006}\tabularnewline
\hline 
\end{tabular}
\label{tab:point estimates}
\end{table}
The inverse power law (L2E-P) is the overall winner in terms of RMSE, although sometimes,
in particular for the random forest (RF), differences are small. In one
case (lasso, $N=100,$ $\nu=1000$), leave-one-out bootstrapping (LOOB) has a slightly
better RMSE than L2E-P. \autoref{tab:point estimates}
also suggest that the inverse power law (parametric) has
a smaller RMSE than the constrained regression spline (nonparametric). 

For the averaged bias, the situation is less straightforward. Learn2Evaluate
always has a positive bias, whereas the competing methods mostly have a negative bias. 
The bootstrap often has the most biased estimate. 
Only for the lasso and the random forest in the $N=100,$ $\nu=1000$ setting and the random forest
in the $N=200,$ $\nu=100$ setting is this not the case. L2E-P, which is always (somewhat) less
biased than L2E-S, is
competitive to 10-fold CV (10F-CV).

A comparison between the learners reveals that the lasso has the
largest RMSE for all simulation settings and predictive
performance estimators, except for the bootstrap in settings $N=100,$ $\nu=1000$
and $N=200,$ $\nu=1000.$ This finding may be explained by
the slope of the learning curve at the end of the learning trajectory.
A large slope causes a drop in predictive performance when it is estimated
on a subset of the data set, as is the case for 10-fold CV and bootstrapping.
Learn2Evaluate also suffers from a larger slope because
extrapolation to $N$ becomes more sensitive to errors. 

\subsection{Confidence bound evaluation}

Coverage probabilities for confidence bounds $L_{n_{opt}}(0.05)$ and $L^{bc}_{n_{opt}}(0.05)$ are given in \autoref{tab:Coverage-results}, 
with $n_{opt}$ determined by MSE minimization. 
We also evaluate the asymptotic 10-fold CV confidence bound estimator for the AUC (Le Dell) and leave-one-out bootstrapping (LOOB). 
We show results for the learning curve fitted by an inverse power law. 
In \hyperref[sec:Supplements]{supplementary Section 3} (Table S1), we show coverage results when $n_{opt}$ is determined by controlling for the empirical bias instead of MSE minimization. 
These coverage results are similar as in \autoref{tab:Coverage-results}.
Results for the constrained regression spline are found in \hyperref[sec:Supplements]{supplementary Section 4} (Table S2).

\begin{table}[h!]
\caption{\footnotesize{Coverage results for the 95\% lower confidence bounds
of Learn2Evaluate with ($L_{n_{opt}}^{bc}$) and without bias correction ($L_{n_{opt}}$), the asymptotic AUC variance
estimator of Le Dell (Le Dell), and leave-one-out bootstrapping (LOOB). The learning curve is fitted by a power law and $n_{opt}$ is determined by MSE minimization.}}
\begin{tabular}{cccccccccc}
\hline 
\multicolumn{2}{c}{} & \textbf{\footnotesize{}Ridge} & \textbf{\footnotesize{}Lasso} & \textbf{\footnotesize{}RF} &  &  & \textbf{\footnotesize{}Ridge} & \textbf{\footnotesize{}Lasso} & \textbf{\footnotesize{}RF}\tabularnewline
\hline 
 & {\footnotesize{}$L_{n_{opt}}$} & {\footnotesize{}0.976} & {\footnotesize{}0.992} & {\footnotesize{}0.980} &  & {\footnotesize{}$L_{n_{opt}}$} & {\footnotesize{}0.974} & {\footnotesize{}0.991} & {\footnotesize{}0.975}\tabularnewline
{\footnotesize{}$\boldsymbol{N=100,}$ } & {\footnotesize{}$L_{n_{opt}}^{bc}$} & {\footnotesize{}0.949} & {\footnotesize{}0.935} & {\footnotesize{}0.937} & {\footnotesize{}$\boldsymbol{N=100,}$ } & {\footnotesize{}$L_{n_{opt}}^{bc}$} & {\footnotesize{}0.941} & {\footnotesize{}0.977} & {\footnotesize{}0.953}\tabularnewline
{\footnotesize{}$\boldsymbol{\nu=1000}$ } & {\footnotesize{}Le Dell} & {\footnotesize{}0.895} & {\footnotesize{}0.851} & {\footnotesize{}0.876} & {\footnotesize{}$\boldsymbol{\nu=100}$ } & {\footnotesize{}Le Dell} & {\footnotesize{}0.870} & {\footnotesize{}0.847} & {\footnotesize{}0.892}\tabularnewline
 & {\footnotesize{}LOOB} & {\footnotesize{}0.990} & {\footnotesize{}0.996} & {\footnotesize{}0.990} &  & {\footnotesize{}LOOB} & {\footnotesize{}0.993} & {\footnotesize{}0.995} & {\footnotesize{}0.995}\tabularnewline
\hline 
 & {\footnotesize{}$L_{n_{opt}}$} & {\footnotesize{}0.992} & {\footnotesize{}0.995} & {\footnotesize{}0.991} &  & {\footnotesize{}$L_{n_{opt}}$} & {\footnotesize{}0.981} & {\footnotesize{}0.996} & {\footnotesize{}0.969}\tabularnewline
{\footnotesize{}$\boldsymbol{N=200,}$ } & {\footnotesize{}$L_{n_{opt}}^{bc}$} & {\footnotesize{}0.965} & {\footnotesize{}0.969} & {\footnotesize{}0.970} & {\footnotesize{}$\boldsymbol{N=200,}$ } & {\footnotesize{}$L_{n_{opt}}^{bc}$} & {\footnotesize{}0.960} & {\footnotesize{}0.974} & {\footnotesize{}0.934}\tabularnewline
{\footnotesize{}$\boldsymbol{\nu=1000}$ } & {\footnotesize{}Le Dell} & {\footnotesize{}0.911} & {\footnotesize{}0.874} & {\footnotesize{}0.884} & {\footnotesize{}$\boldsymbol{\nu=100}$ } & {\footnotesize{}Le Dell} & {\footnotesize{}0.908} & {\footnotesize{}0.886} & {\footnotesize{}0.869}\tabularnewline
 & {\footnotesize{}LOOB} & {\footnotesize{}0.996} & {\footnotesize{}1,00} & {\footnotesize{}0.992} &  & {\footnotesize{}LOOB} & {\footnotesize{}0.996} & {\footnotesize{}0.998} & {\footnotesize{}0.989}\tabularnewline
\hline 
\end{tabular}
\label{tab:Coverage-results} 
\end{table}

\autoref{tab:Coverage-results} demonstrates that $L_{n_{opt}}(0.05)$ controls the type-\mbox{I}
error well below the nominal level, while being less conservative
than leave-one-out bootstrapping. Coverage is closer to nominal
level for $L^{bc}_{n_{opt}}(0.05)$.
In one instance (random forest, $N=200,$ $\nu=100$, $L^{bc}_{n_{opt}}$),
the coverage drops just below 95\% when taking the binomial error
of the coverage estimate into account. When the learning curve is fitted by a constrained regression
spline, $L^{bc}_{n_{opt}}(0.05)$ is too liberal in three
cases (\hyperref[sec:Supplements]{supplementary Table S2}). This result is explained by the worse point estimates of the spline 
compared to the power law.  The asymptotic confidence bound estimator (Le Dell) is too liberal for small sample sizes, similar as
in \cite{LeDell2015ComputationallyEC}.

In \hyperref[sec:Supplements]{supplementary Section 5} (Figures S2, S3, S4, and S5), we depict
boxplots of the distance of the lower confidence bound to the corresponding target parameter $\mbox{AUC}_{r}$ 
for the bootstrap, and Learn2Evaluate with ($L^{bc}_{n_{opt}}$) and without ($L_{n_{opt}}$) bias correction.
These boxplots show that $L_{n_{opt}}$ and $L^{bc}_{n_{opt}}$ are on average closer to the target parameter $\mbox{AUC}_{r}$ than bootstrapping, 
with an exception of $L_{n_{opt}}$ for the lasso. 
The bias correction shortens the distance compared to the confidence bound without bias correction.

The coverage results and the boxplots also illustrate that $L_{n_{opt}}$
for lasso regression is further away from $\mbox{AUC}_{r}$ than for 
ridge regression and random forest. The average optimal training set size $n_{opt}$ for lasso regression
is larger than for ridge regression and random forest (Appendix \autoref{tab:nopt power law}). 
Hence, the average test set size for lasso regression is smaller, which leads to wider confidence bounds.
For $L^{bc}_{n_{opt}}$, the difference between the learners is smaller.

\section{Applications}
\label{sec:Applications}

\subsection{Classification}
\label{sec:Classification}

We apply Learn2Evaluate to messenger-RNA sequencing data, 
extracted from blood platelets, as described in \cite{best2015rna}.
These experiments demonstrated that RNA profiles of blood platelets 
are a promising diagnostic tool for early cancer detection.

RNA profiles of blood platelets from $230$ patients, having one of
the in total six tumor types, and $55$ healthy controls were obtained.
The raw data are online available in the GEO database (\href{https://www.ncbi.nlm.nih.gov/geo/query/acc.cgi?acc=GSE68086}{GEO: GSE68086}).
Data preprocessing is described in \cite{novianti2017better}. Processed data are available via
\url{https://github.com/JeroenGoedhart/Learn2Evaluate} as well as the R code 
to apply Learn2Evaluate to these data.  For additional details on the data 
and the code, we refer to \hyperref[sec:Supplements]{supplementary Section 9}.

Here, we consider three binary classification cases: non-small-cell
lung cancer $(\mbox{NSCLC,}\quad n=60)$ versus the control
group $(n=55)$ using $p=19,300$ transcripts (\autoref{fig:NSCLC vs HC}), breast cancer $(\mbox{Breast,}\quad n=40)$ versus
NSCLC using $p=18,741$ transcripts (\autoref{fig:Breast vs NSCLC}),
and breast cancer versus pancreas cancer $(\mbox{Pancreas,}\quad n=35)$ using
$p=18,350$ transcripts (\autoref{fig:Breast vs Pancreas}). For
three other binary classification experiments, we refer to \hyperref[sec:Supplements]{supplementary Section 6} 
(Figure S6).

Next, we generate learning curves for ridge regression (black), lasso
regression (blue), and random forest (green). Learning curves are
obtained by the same procedure as in the simulation study. For ten
homogeneously spaced subsample sizes, we estimate the predictive performance,
quantified by the AUC, by repeated ($50$ times) hold-out estimation.
This renders the learning trajectory (dots), i.e. \eqref{eq:learning trajectory}.
The maximum subsample size equals $N-10.$ We varied the number
of different subsamples and repeats to establish that these defaults
yield stable estimates (\hyperref[sec:Supplements]{supplementary Section 8}, Table S3). We then fit an inverse
power law given by \eqref{eq:inverse power law} (solid line) because this 
fit obtained the best performance estimates in the simulation.

We construct confidence bounds based on Learn2Evaluate without bias correction ($L_{n_{opt}}$, upward triangle), Learn2Evaluate with bias correction ($L^{bc}_{n_{opt}}$, downward triangle), and leave-one-out bootstrapping (Bootstrap, square).
We use MSE minimization, i.e. \eqref{eq: nopt mse}, to determine
an optimal training subsample size $n_{opt}$ (dashed vertical line). We also
display the point estimates of Learn2Evaluate (star).
\begin{figure}[t!]
    \centering
    
    \begin{subfigure}[t]{0.49\textwidth}
        \centering
        \includegraphics[bb=0bp 10bp 504bp 412bp,width=1\textwidth]{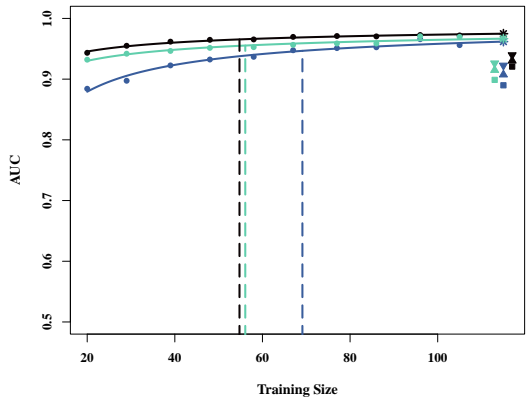}
        \caption{\footnotesize{Control vs. NSCLC}}
        \label{fig:NSCLC vs HC}
    \end{subfigure}
    \hfill
    \begin{subfigure}[t]{0.49\textwidth}
        \centering
	\includegraphics[bb=0bp 10bp 504bp 412bp,width=1\textwidth]{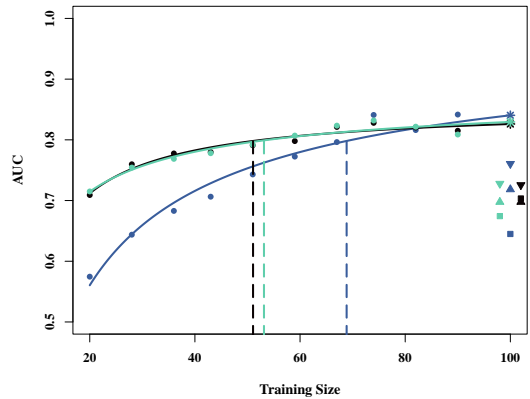}    
	\caption{\footnotesize{Breast vs. NSCLC}}
        \label{fig:Breast vs NSCLC}
    \end{subfigure}\\
     \begin{subfigure}[t]{0.49\textwidth}
        \centering
	\includegraphics[bb=0bp 10bp 504bp 412bp,width=1\textwidth]{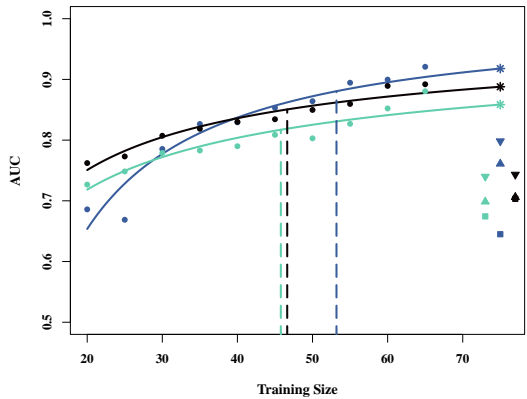}    
	\caption{\footnotesize{Breast vs. Pancreas}}
        \label{fig:Breast vs Pancreas}
     \end{subfigure}
     \hfill
     \begin{subfigure}[t]{0.49\textwidth}
        \centering
	\captionsetup[subfigure]{labelformat=empty}
	\includegraphics[bb=0bp 10bp 504bp 412bp,width=1\textwidth]{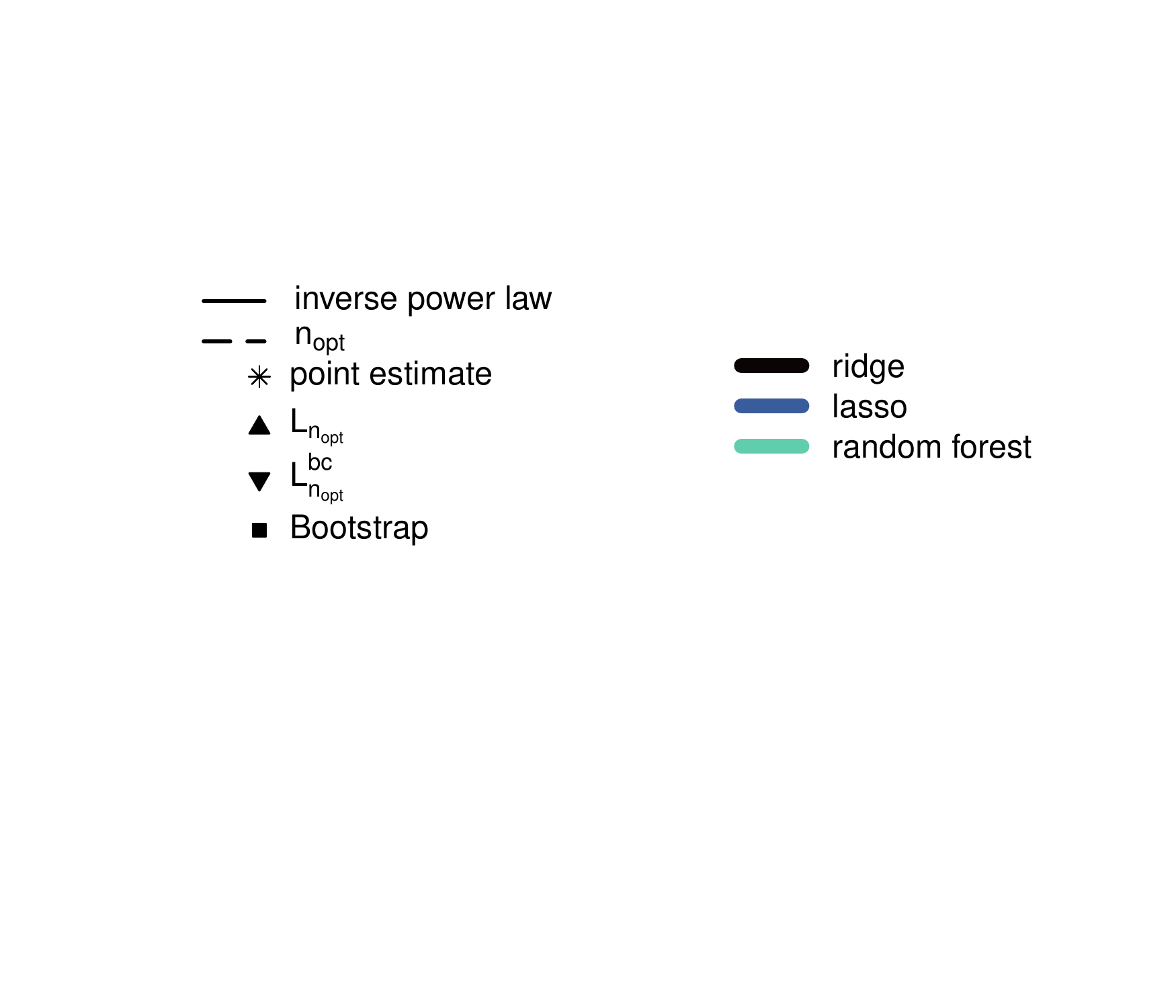}    
	\label{fig:legend}
    \end{subfigure}
     \caption{\footnotesize{Learning curves for three classification
	tasks. In each task, we produce a learning curve for a ridge (black),
	a lasso (blue), and a random forest (green) model. For each learner, we depict the learning trajectory (dots)
	together with the inverse power law fit (solid line).
	We show confidence bounds of Learn2Evaluate with  ($L^{bc}_{n_{opt}}$,
	downward triangle) and without ($L_{n_{opt}}$, upward triangle) bias correction, and bootstrapping (Bootstrap, square). We use MSE minimization to 
       	determine $n_{opt}$ (dashed vertical line).}}
    \label{fig:Learning-curves}
\end{figure}
\autoref{fig:Learning-curves} and \hyperref[sec:Supplements]{supplementary Figure S6} illustrate that 
the theoretically justified confidence bound of Learn2Evaluate ($L_{n_{opt}}$) is, in most cases, slightly
tighter than the bootstrap confidence bound, which agrees 
with the simulation results. A bias correction ($L^{bc}_{n_{opt}}$)
tightens the theoretically justified bound as expected. Furthermore, \autoref{fig:Learning-curves}
shows that $n_{opt}$ adapts to the saturation level of the
learning curve. A more saturated curve leads to a smaller training
set and hence to a larger test set to determine the lower bound.

\autoref{fig:Learning-curves} and \hyperref[sec:Supplements]{supplementary Figure S6}
also illustrate the additional benefits of using learning curves
for predictive performance estimation. Firstly, the curves indicate
how future samples increase the AUC. For lasso regression, which has
in most cases a large slope, additional samples are expected to increase
the predictive performance. The ridge and random forest level off completely in
\autoref{fig:NSCLC vs HC}, suggesting that future samples do not increase the AUC. 

Secondly, the curves provide a more complete comparison between learners
than just employing single point estimates. For example, point estimates
of the lasso and the ridge are almost similar in \autoref{fig:Breast vs NSCLC}.
The lasso model may still be preferred because its learning curve
has a larger slope, which likely leads to a larger increase in the AUC when additional training samples are available.

Finally, the curves show different learning rates. The lasso model
tends to be a slow learner. At small subsample sizes, it performs
much worse than the ridge and random forest model, but it catches
up at larger subsample sizes.

\subsection{Regression}
\label{sec:Regression}
The level of methylation of CpG islands  has been shown to be predictive of age \citep{DNAmeth}. 
Here, we apply Learn2Evaluate to a DNA methylation data set ($N=108$ and $p=2289$), which is 
used to accurately predict age. The predictive performance is quantified by the PMSE. Details on this data set
and the learning curves are found in \hyperref[sec:Supplements]{supplementary Section 7} (Figure S7).

Figure S7 illustrates that the decreasing and convex power law fits the learning trajectories for the PMSE well.
The upper confidence bounds of Learn2Evaluate are in most cases tighter than the bootstrapped confidence bounds. Only the theoretically justified 
confidence bound $L_{n_{opt}}$ of random forest is more conservative. The bias-corrected confidence 
bounds $L^{bc}_{n_{opt}}$ are always tighter than the bootstrapped confidence bounds. Again, the optimal 
training size $n_{opt}$ adapts to the saturation level of the curves. 

Lasso regression has a smaller PMSE than ridge regression and random forest in this application. 
Interestingly, the bias-corrected confidence bound ($L^{bc}_{n_{opt}}$) of lasso regression is lower valued
than the point estimates of ridge regression and random forest, which suggests a clear difference between the learners.

\section{Discussion}
\label{sec:Discussion}

We presented a novel method, called Learn2Evaluate, to estimate the predictive performance. 
Learn2Evaluate facilitates the computation of a point estimate and a lower confidence bound, 
which is proven to control type-\mbox{I} error.
This bound may be further tightened by
correcting for the estimated bias. This bias-corrected bound is shown to approximately control 
type-\mbox{I} error in a simulation, although the nominal confidence level is not always guaranteed. 
Simulations and applications to omics data showed that both bounds are less conservative than a 
bootstrapped confidence bound. Learn2Evaluate appears to have a lower RMSE of performance point estimates 
than 10-fold cross-validation and leave-one-out bootstrapping. 
Furthermore, Learn2Evaluate comes with some additional benefits by providing 
a dynamic comparison between learners and insight in the potential benefit of future samples.

One practical limitation of Learn2Evaluate is its computational time
because of the relatively large number of splits in a training and
test set. This may be a drawback in medium to large
sample size settings, because computationally more efficient (asymptotic)
methods often suffice \citep{LeDell2015ComputationallyEC}. However,
for small sample sizes, other resampling techniques also need to generate a large
number of splits in a training and test set to yield a reliable performance estimator 
\citep{jiang2007comparison,jiang2008calculating,kim2009repeatedcv}.

In the simulations and the applications, we focused on high dimensional data settings with sample sizes
ranging from $N=75$ to $N=200$. We focused on omics data 
by estimating the covariate correlation structure from an omics experiment and by
imposing a dense structure on the parameter vector $\boldsymbol{\beta}.$
However, Learn2Evaluate applies to any data setting and
covariate-response structure. Evaluating the practical usefulness
of Learn2Evaluate for such generalizations requires additional simulations.

Results in this study suggest that the power law, which uses all subsample sizes equally, obtains better 
point estimates than the regression spline, which predominantly uses training
set sizes at the end of the learning trajectory to extrapolate. It is therefore interesting
to investigate which subsample sizes are most informative for estimating the performance at the full sample size.
This may lead to an optimal weighting scheme in the nonlinear least squares fit by an inverse power law, which may improve our point estimates further.
Finding an optimal weighting scheme, however, is nontrivial. The variance of the repeated hold-out estimates is difficult to asses because it depends on the complex correlation structure between overlapping training
sets \citep{NoUnbiased}. We therefore leave this for future research.

\section*{Declaration of Interest}
\label{sec:Declaration of Interest}

The authors declare no competing interests.

\section*{Acknowledgement}
\label{sec:Acknowledgement}

The authors gratefully acknowledge the financial support by Stichting Hanarth Fonds.

\appendix
\setcounter{table}{0}
\renewcommand{\thetable}{\Alph{section}\arabic{table}}
\section{Average True AUC}
\label{sec:True Average AUC}
\begin{table}[H]
\centering{}\caption{\footnotesize{Average true AUC for the different simulation settings and learners}}
\begin{tabular}{cccc}
\hline 
\multicolumn{1}{c}{} & \textbf{\footnotesize{}Ridge} & \textbf{\footnotesize{}Lasso} & \textbf{\footnotesize{}RF}\tabularnewline
\hline 
\multirow{1}{*}{{\footnotesize{}$N=100,$ $\nu=1000$}} & {\footnotesize{}0.77} & {\footnotesize{}0.71} & {\footnotesize{}0.75}\tabularnewline
\multirow{1}{*}{{\footnotesize{}$N=100,$ $\nu=100$}} & {\footnotesize{}0.88} & {\footnotesize{}0.83} & {\footnotesize{}0.87}\tabularnewline
\multirow{1}{*}{{\footnotesize{}$N=200,$ $\nu=1000$}} & {\footnotesize{}0.78} & {\footnotesize{}0.74} & {\footnotesize{}0.76}\tabularnewline
\multirow{1}{*}{{\footnotesize{}$N=200,$ $\nu=100$}} & {\footnotesize{}0.89} & {\footnotesize{}0.86} & {\footnotesize{}0.87}\tabularnewline
\hline 
\end{tabular}
\label{tab:True-AUC-results}
\end{table}
\setcounter{table}{0}
\section{Average Training Set Size }
\label{sec:Training set size}
\begin{table}[H]
\centering{}
\caption{\footnotesize{Average $n_{opt}$ for ridge regression, lasso regression, and random forest (RF).}}
\begin{tabular}{cccc}
\hline 
 & \multicolumn{1}{c}{\textbf{\footnotesize{}Ridge}} & \multicolumn{1}{c}{\textbf{\footnotesize{}Lasso}} & \multicolumn{1}{c}{\textbf{\footnotesize{}RF}}\tabularnewline
\hline 
\multirow{1}{*}{{\footnotesize{}$N=100,$ $\nu=1000$}} & {\footnotesize{}32 } & {\footnotesize{}45 } & {\footnotesize{}33 }\tabularnewline
\multirow{1}{*}{{\footnotesize{}$N=100,$ $\nu=100$}} & {\footnotesize{}36 } & {\footnotesize{}54 } & {\footnotesize{}34}\textbf{ }\tabularnewline
\multirow{1}{*}{{\footnotesize{}$N=200,$ $\nu=1000$}} & {\footnotesize{}69 } & {\footnotesize{}101 } & {\footnotesize{}76 }\tabularnewline
\multirow{1}{*}{{\footnotesize{}$N=200,$ $\nu=100$}} & {\footnotesize{}76}  & {\footnotesize{}109 } & {\footnotesize{}69}\textbf{ }\tabularnewline
\hline 
\end{tabular}
\label{tab:nopt power law} 
\end{table}

\section{Supplementary material}
\label{sec:Supplements}

Section 1 presents a method to deal with S-shaped learning curves, 
which may arise as a consequence of a phase transition. Figure S1 shows 
an instance of a S-shaped learning curve. In Section 2, we show asymptotic 
variance estimators for the AUC and the PMSE. If such an estimator is not available 
for the given performance metric, we describe an alternative method in Section 3.
Table S1 shows coverage results for this alternative and in Section 4, coverage results for 
the learning curve fitted by a constrained regression spline are presented. In Section 5, we show boxplots depicting 
the distance of the confidence bounds of Learn2Evaluate to the true performance (Figures S2, S3, S4, and S5).
Section 6 deals with an additional classification application (Figure S6) and in Section 7, Learn2Evaluate is applied in 
a regression setting. Table S3 (Section 8) shows the stability of the point estimates of Learn2Evaluate with respect to the length of the learning trajectory.
Finally, Section 9 gives additional details on the data and the software.

\end{document}


\sloppy

\def\spacingset#1{\renewcommand{\baselinestretch}%
{#1}\small\normalsize} \spacingset{1}

\title{SUPPLEMENTARY MATERIAL TO: Estimation of Predictive Performance in High-Dimensional Data Settings using Learning Curves}

 \author{Jeroen M. Goedhart\footnote{Corresponding author, E-mail address: j.m.goedhart@amsterdamumc.nl} \textsuperscript{a}, Thomas Klausch\textsuperscript{a}, Mark A. van de Wiel\textsuperscript{a}\\
\textsuperscript{a}\small{Department of Epidemiology and Data Science, Amsterdam University Medical Centers}}
\maketitle
\spacingset{1.5}
\section{S-shaped learning curves}
\label{sec:S-shaped learning curves}

In some situations, the concavity or convexity constraint on the learning
curve may be inadequate. An example of such a situation is the phase
transition of lasso regression when the true covariate-response structure
is sparse and the outcome continuous \citep{PhaseTransitionDonoho,DonohoPhaseTransition}.
Here, the performance, quantified by the mean square error of prediction (PMSE), follows a reverse S-shaped
curve as a function of the training set size. In the first phase,
the curve is decreasing and concave, and in the second phase, starting
after the inflection point, the curve is decreasing and convex. 

For completeness, we extend Learn2Evaluate to also deal with S-shaped
curves. Users should visually inspect whether such a curve fit is
appropriate by plotting the learning trajectory (eq. 2 in main document). For phase
transitions, we are mainly interested in the second phase because
the curves are used to extrapolate to the full sample size $N$. Therefore,
we start by estimating the inflection point $n_{tr},$ after which
the curve will have its familiar behavior: increasing and concave for increasing performance metrics such as the AUC, and decreasing and convex for 
decreasing performance metrics such as the mean square error of prediction (PMSE).
To determine $n_{tr},$ we fit a generalized logistic curve \citep{RichardCurve},
which is defined by 
\begin{equation}
f(n)=A+\frac{K-A}{\left(1+\exp\left(-B(n-M\right)\right)^{1/\nu}}. \label{eq:S_curve}
\end{equation}
with $A$ being the left asymptote, $K$ the right asymptote, and
$B$ the growth rate. Parameters $M$ and $\nu$ are related to the
(a)symmetry of the curve. Note that \eqref{eq:S_curve} holds for both decreasing and increasing performance metrics by modifying the parameters. To estimate the parameters, we solve 
\begin{equation}
(\hat{A},\hat{B},\hat{K},\hat{M},\hat{\nu})=\argmin_{A,B,K,M,\nu}\frac{1}{J}\sum_{j=1}^{J}\left(A+\frac{K-A}{\left(1+\exp\left(-B(n_{j}-M\right)\right)^{1/\nu}}-\hat{\rho}_{n_{j}}\right)^{2}, \label{eq:S_curve_fitting}
\end{equation}
using the base R optim function. The estimated inflection point $\hat{n}_{tr}$
is then given by
\begin{equation}
\hat{n}_{tr}=\hat{M}-\frac{1}{\hat{B}}\log\hat{\nu},\label{eq:IP_point}
\end{equation}
as derived in \citep{IP_RichardCurve}. After $\hat{n}_{tr}$ is determined,
we fit the standard learning curve using $\hat{n}_{tr}$ as minimal
subsample size, i.e. $n_{1}=\hat{n}_{tr}.$

\subsection{Illustration on simulated data}

Here, we illustrate the above described methodology on simulated data for lasso regression and the PMSE as performance metric. 

The $i$th instance of response variable $Y$ is defined by
\begin{equation}
Y_{i}=\boldsymbol{X}_{i}\boldsymbol{\beta}+\epsilon,\qquad i=1,\ldots,N,
\end{equation}
with $\epsilon\sim\mathcal{N}\left(0,0.2\right)$ and $N=100.$ The
$p=1000$ covariates $\boldsymbol{X}_{i}$ are independently drawn
from a standard normal distribution, $X_{i}\sim\mathcal{N}\left(0,1\right).$
For the sparse parameter vector $\boldsymbol{\beta},$ we set the first five
entries to $0.6$ and the remaining entries are fixed at zero. 

We construct learning curves (\autoref{fig:Phase-transition}) for lasso regression by 
first estimating a learning trajectory of
$20$ different subsample sizes $n_{j}$ in the regime $[10,90]$ by applying repeated
($50$ repeats) hold-out estimation (blue dots). The penalty parameter of lasso regression is estimated once per $n_{j},$ which is then fixed for the $50$ repeats of each $n_{j}.$.
We estimate this penalty parameter by the median of repeated (5 times) 10-fold CV. 
We fit the learning trajectory using \eqref{eq:S_curve_fitting} (blue solid line) and we estimate the inflection point
$n_{tr}$ using \eqref{eq:IP_point} (black dotted vertical line). After $\hat{n}_{tr}$
is determined, we estimate the learning trajectory in the subsample
regime $[\hat{n}_{tr},90]$ (black dots)
using the same procedure as above. We fit this learning trajectory
using the inverse power law for decreasing and convex data (black solid line).

\begin{figure}[H]
\begin{center}
\includegraphics[bb=3bp 17bp 475bp 250bp,clip]{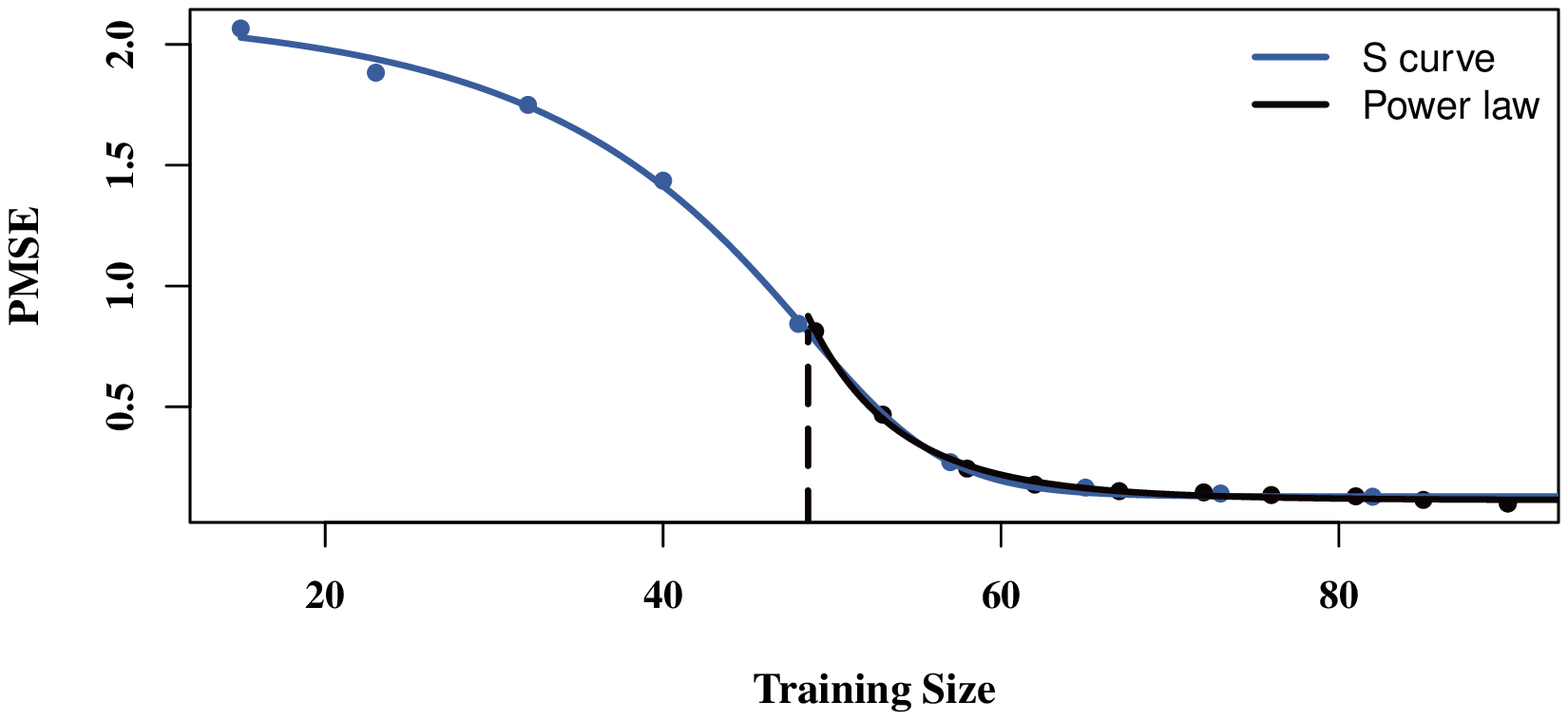}
\caption{\footnotesize{Phase transition for lasso regression.
The black dotted vertical line indicates the inflection point of the
generalized logistic curve.}}
\label{fig:Phase-transition}
\end{center}
\end{figure}

\autoref{fig:Phase-transition} illustrates that the generalized logistic function, i.e. \eqref{eq:S_curve}, describes phase transitions for lasso regression well.
Furthermore, the inverse power law and  \eqref{eq:S_curve} show good agreement after the estimated inflection point $\hat{n}_{tr}.$ 

Although in this case the differences between \eqref{eq:S_curve} and the decreasing convex inverse power law are negligible, we prefer to fit a power law after the estimated inflection point. 
The second phase of the learning trajectory is of interest to estimate the performance and a training set size for the confidence bound. In this way, the potential influence
of the first phase of the learning trajectory on these estimates is avoided.

\section{ Asymptotic variance estimators}
\label{sec:Asymptotic variance estimator}

To minimize the mean square error of the predictive performance estimate with respect to the
subsample size $n$ (eq. 11 of the main document), an asymptotic variance estimator is required.
Here, we show asymptotic variance estimators for the AUC and the PMSE. 

For the AUC, \cite{bamber1975area} derived an estimator, which is given by 
\begin{equation}
\widehat{\text{Var}}_{n}(AUC)=\frac{Q_{1}+Q_{2}+Q_{3}}{n_{+}n_{-}},\label{eq: AUC variance esitmator}
\end{equation}
with
\begin{multline*}
\begin{aligned} & Q_{1}=\widehat{AUC}_{n}(1-\widehat{AUC}_{n}),\quad  Q_{2}=\left(n_{+}-1\right)\left(\frac{\hat{\rho}_{n}}{2-\widehat{AUC}_{n}}-\widehat{AUC}_{n}^{2}\right),\quad \mbox{and}\\
 & Q_{3}=(n_{-}-1)\left(\frac{2\widehat{AUC}_{n}^{2}}{1+\widehat{AUC}_{n}}-\widehat{AUC}_{n}^{2}\right).
\end{aligned}
\end{multline*}
Here, $n_{+}$ denotes the number of positive cases ($Y=1$) and $n_{-}=n-n_{+}$
the number of negative cases ($Y=0$) in the test set of size $n$. 

For the PMSE, we use the variance estimator derived by \cite{FABER199979}, which reads 
\begin{equation}
\widehat{\text{Var}}_{n}(\mbox{PMSE})=\frac{2\widehat{\mbox{PMSE}}_{n}^{2}}{n}. \label{eq:PMSEVariance}
\end{equation}

\section{Alternative method to MSE minimization}
\label{sec:Biascontrol}

\subsection{Method}

To determine an optimal training set size $n_{opt}$ when an asymptotic performance variance estimator such as \eqref{eq: AUC variance esitmator} or \eqref{eq:PMSEVariance} is not known, 
we suggest to control for the (empirical) bias.  Specifically, we define an acceptable margin $m$ 
and we find the smallest subsample size $n$ for which we have
\begin{equation}
n_{opt}=\argmin_{n}\left\{ n:\:f(N)-f(n)\leq m\right\} .\label{eq: nopt bias}
\end{equation}
This method, which we call bias margin control, ensures that we obtain a maximum amount of test samples, 
while controlling the empirical bias. A drawback of this method is that the choice of $m$ is arbitrary.

\subsection{Simulation results}

\autoref{tab:Coverage-results-1} shows coverage results for the $95\%$ confidence bounds of Learn2Evaluate with ($L_{n_{opt}}^{bc}$)  and without ($L_{n_{opt}}$) bias correction. 
Here, we determine $n_{opt}$ by bias margin control, i.e. \eqref{eq: nopt bias}, using a margin of $2\%$ of the point estimate $f(N).$ For completeness, we also show coverage results for the asymptotic estimator derived by \cite{LeDell2015ComputationallyEC} (Le Dell) and leave-one-out bootstrapping (LOOB).

\begin{table}[H]
\caption{\footnotesize{Coverage results for the 95\% lower confidence bound
of Learn2Evaluate without bias correction ($L_{n_{opt}}$), Learn2Evaluate
with bias correction ($L_{n_{opt}}^{bc}$), the asymptotic AUC variance
estimator of Le Dell (Le Dell), and leave-one-out bootstrapping (LOOB).
We use bias margin control to determine $n_{opt}.$  Learning curves are fitted by inverse power laws.}}
\begin{tabular}{cccccccccc}
\hline 
\multicolumn{2}{c}{} & \textbf{\footnotesize{}Ridge} & \textbf{\footnotesize{}Lasso} & \textbf{\footnotesize{}RF} &  &  & \textbf{\footnotesize{}Ridge} & \textbf{\footnotesize{}Lasso} & \textbf{\footnotesize{}RF}\tabularnewline
\hline 
 & {\footnotesize{}$L_{n_{opt}}$} & {\footnotesize{}0.974} & {\footnotesize{}0.988} & {\footnotesize{}0.979} &  & {\footnotesize{}$L_{n_{opt}}$} & {\footnotesize{}0.976} & {\footnotesize{}0.985} & {\footnotesize{}0.983}\tabularnewline
{\footnotesize{}$\boldsymbol{N=100,}$} & {\footnotesize{}$L_{n_{opt}}^{bc}$} & {\footnotesize{}0.954} & {\footnotesize{}0.965} & {\footnotesize{}0.935} & {\footnotesize{}$\boldsymbol{N=100,}$} & {\footnotesize{}$L_{n_{opt}}^{bc}$} & {\footnotesize{}0.937} & {\footnotesize{}0.977} & {\footnotesize{}0.945}\tabularnewline
{\footnotesize{}$\boldsymbol{\nu=1000}$} & {\footnotesize{}Le Dell} & {\footnotesize{}0.895} & {\footnotesize{}0.851} & {\footnotesize{}0.876} & {\footnotesize{}$\boldsymbol{\nu=100}$} & {\footnotesize{}Le Dell} & {\footnotesize{}0.870} & {\footnotesize{}0.847} & {\footnotesize{}0.892}\tabularnewline
 & {\footnotesize{}LOOB} & {\footnotesize{}0.990} & {\footnotesize{}0.996} & {\footnotesize{}0.990} &  & {\footnotesize{}LOOB} & {\footnotesize{}0.993} & {\footnotesize{}0.995} & {\footnotesize{}0.995}\tabularnewline
\hline 
 & {\footnotesize{}$L_{n_{opt}}$} & {\footnotesize{}0.992} & {\footnotesize{}0.993} & {\footnotesize{}0.992} &  & {\footnotesize{}$L_{n_{opt}}$} & {\footnotesize{}0.994} & {\footnotesize{}0.995} & {\footnotesize{}0.979}\tabularnewline
{\footnotesize{}$\boldsymbol{N=200,}$} & {\footnotesize{}$L_{n_{opt}}^{bc}$} & {\footnotesize{}0.972} & {\footnotesize{}0.985} & {\footnotesize{}0.971} & {\footnotesize{}$\boldsymbol{N=200,}$} & {\footnotesize{}$L_{n_{opt}}^{bc}$} & {\footnotesize{}0.952} & {\footnotesize{}0.973} & {\footnotesize{}0.928}\tabularnewline
{\footnotesize{}$\boldsymbol{\nu=1000}$} & {\footnotesize{}Le Dell} & {\footnotesize{}0.911} & {\footnotesize{}0.874} & {\footnotesize{}0.884} & {\footnotesize{}$\boldsymbol{\nu=100}$} & {\footnotesize{}Le Dell} & {\footnotesize{}0.908} & {\footnotesize{}0.886} & {\footnotesize{}0.869}\tabularnewline
 & {\footnotesize{}LOOB} & {\footnotesize{}0.996} & {\footnotesize{}1,00} & {\footnotesize{}0.992} &  & {\footnotesize{}LOOB} & {\footnotesize{}0.996} & {\footnotesize{}0.998} & {\footnotesize{}0.989}\tabularnewline
\hline 
\end{tabular}
\label{tab:Coverage-results-1}
\end{table}

\noindent \autoref{tab:Coverage-results-1} shows similar results as in Table 2 of the main document. 
The theoretically justified lower confidence bound $L_{n_{opt}}$ controls the type-\mbox{I}
error well below the nominal level and is less conservative than leave-one-out bootstrapping. 
Coverage is closer to nominal level when a bias correction is performed ($L^{bc}_{n_{opt}}$).
Again, in setting $N=200,$ $\nu=100,$ $L^{bc}_{n_{opt}}$ of random forest is too liberal 
when taking the binomial error of the coverage estimate into account.

\section{Coverage results for the spline}
\label{sec:coverage spline}

\noindent\autoref{tab:Coverage-results-regression spline} shows the coverage of the 
$95\%$ lower confidence bounds $L_{n_{opt}}$ and $L^{bc}_{n_{opt}}$ of Learn2Evaluate, with learning curves fitted by a constrained regression spline.
We show results for both methods to determine an optimal training set size $n_{opt},$ i.e. bias margin control (Bias, \autoref{eq: nopt bias}) and MSE minimization (MSE, eq. 11 of the main document).

\begin{table}[H]
\centering{}
\caption{\footnotesize{Coverage results for the 95\% lower confidence bound
of Learn2Evaluate without bias correction ($L_{n_{opt}}$) and Learn2Evaluate
with bias correction ($L_{n_{opt}}^{bc}$). We use bias margin control (Bias) and MSE minimization
(MSE) to determine $n_{opt}.$ Learning curves are fitted by constrained regression splines.}}
\begin{tabular}{cc>{\centering}p{0.9cm}>{\centering}p{0.9cm}>{\centering}p{0.9cm}cc>{\centering}p{0.9cm}>{\centering}p{0.9cm}>{\centering}p{0.9cm}}
\hline 
\multicolumn{2}{c}{} & \textbf{\footnotesize{}Ridge}{\footnotesize{} } & \textbf{\footnotesize{}Lasso}{\footnotesize{} } & \textbf{\footnotesize{}RF}{\footnotesize{} } &  &  & \textbf{\footnotesize{}Ridge}{\footnotesize{} } & \textbf{\footnotesize{}Lasso}{\footnotesize{} } & \textbf{\footnotesize{}RF}\tabularnewline
\hline 
 & {\footnotesize{}$L_{n_{opt}}$ Bias}  & {\footnotesize{}0.977}  & {\footnotesize{}0.987}  & {\footnotesize{}0.983}  &  & {\footnotesize{}$L_{n_{opt}}$ Bias}  & {\footnotesize{}0.978}  & {\footnotesize{}0.983}  & {\footnotesize{}0.982}\tabularnewline
{\footnotesize{}$\boldsymbol{N=100,}$}  & {\footnotesize{}$L_{n_{opt}}$} {\footnotesize{} MSE}  & {\footnotesize{}0.976}  & {\footnotesize{}0.979}  & {\footnotesize{}0.983}  & {\footnotesize{}$\boldsymbol{N=100,}$}  & {\footnotesize{}$L_{n_{opt}}$} {\footnotesize{} MSE}  & {\footnotesize{}0.978}  & {\footnotesize{}0.959}  & {\footnotesize{}0.985}\tabularnewline
{\footnotesize{}$\boldsymbol{\nu=1000}$}  & {\footnotesize{}$L_{n_{opt}}^{bc}$ Bias}  & {\footnotesize{}0.938}  & {\footnotesize{}0.925}  & {\footnotesize{}0.923}  & {\footnotesize{}$\boldsymbol{\nu=100}$}  & {\footnotesize{}$L_{n_{opt}}^{bc}$ Bias}  & {\footnotesize{}0.926}  & {\footnotesize{}0.961}  & {\footnotesize{}0.942}\tabularnewline
 & {\footnotesize{}$L_{n_{opt}}^{bc}$ MSE}  & {\footnotesize{}0.961}  & {\footnotesize{}0.952}  & {\footnotesize{}0.945}  &  & {\footnotesize{}$L_{n_{opt}}^{bc}$ MSE}  & {\footnotesize{}0.938}  & {\footnotesize{}0.950}  & {\footnotesize{}0.942}\tabularnewline
\hline 
 & {\footnotesize{}$L_{n_{opt}}$} {\footnotesize{} Bias}  & {\footnotesize{}0.991}  & {\footnotesize{}0.993}  & {\footnotesize{}0.990}  &  & {\footnotesize{}$L_{n_{opt}}$ Bias}  & {\footnotesize{}0.985}  & {\footnotesize{}0.995}  & {\footnotesize{}0.974}\tabularnewline
{\footnotesize{}$\boldsymbol{N=200,}$}  & {\footnotesize{}$L_{n_{opt}}$ MSE}  & {\footnotesize{}0.996}  & {\footnotesize{}0.997}  & {\footnotesize{}0.994}  & {\footnotesize{}$\boldsymbol{N=200,}$}  & {\footnotesize{}$L_{n_{opt}}$ MSE}  & {\footnotesize{}0.991}  & {\footnotesize{}0.991}  & {\footnotesize{}0.983}\tabularnewline
{\footnotesize{}$\boldsymbol{\nu=1000}$}  & {\footnotesize{}$L_{n_{opt}}^{bc}$ Bias}  & {\footnotesize{}0.959}  & {\footnotesize{}0.961}  & {\footnotesize{}0.970}  & {\footnotesize{}$\boldsymbol{\nu=100}$}  & {\footnotesize{}$L_{n_{opt}}^{bc}$ Bias}  & {\footnotesize{}0.957}  & {\footnotesize{}0.905}  & {\footnotesize{}0.923}\tabularnewline
 & {\footnotesize{}$L_{n_{opt}}^{bc}$ MSE}  & {\footnotesize{}0.976}  & {\footnotesize{}0.987}  & {\footnotesize{}0.975}  &  & {\footnotesize{}$L_{n_{opt}}^{bc}$ MSE}  & {\footnotesize{}0.952}  & {\footnotesize{}0.969}  & {\footnotesize{}0.925}\tabularnewline
\hline 
\end{tabular}\label{tab:Coverage-results-regression spline} 
\end{table}

\section{Boxplots of confidence bound distances}
\label{sec:boxplots cb}

Here, we depict, for all simulation settings, boxplots of the distance of the
lower confidence bound to the corresponding target parameter $\mbox{AUC}_{r}$
for the bootstrap method (gray), and Learn2Evaluate with (L2E + BC, brown),
and without (L2E - BC, green) bias correction. The learning curve is fitted
by an inverse power law (eq. 4 of the main document) and
$n_{opt}$ is determined by MSE minimization (eq. 11 of the main document).

\begin{figure}[H]
\begin{center}
\includegraphics[bb=0bp 50bp 550bp 305bp,clip,width=6in]{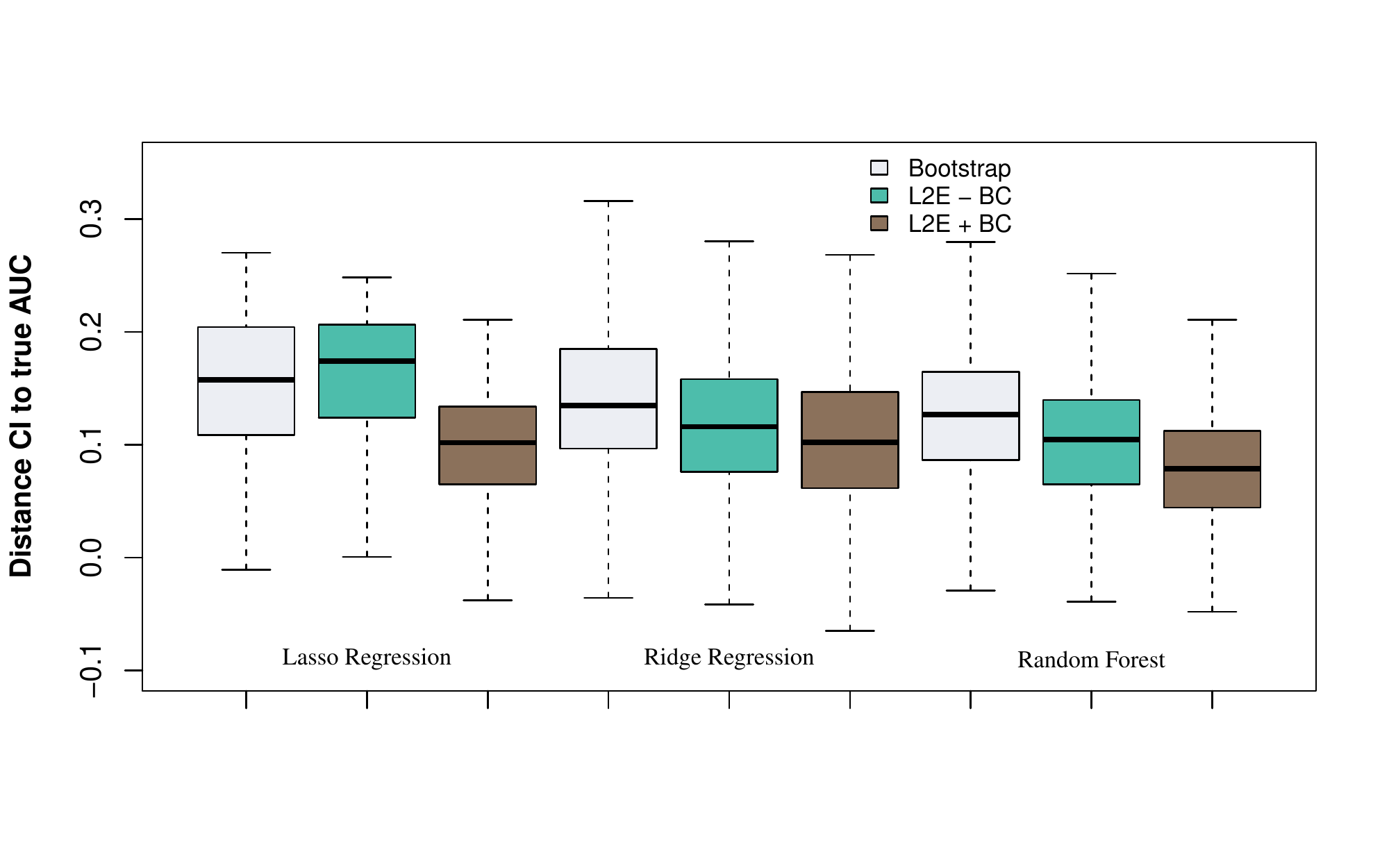}
\end{center}
\caption{\footnotesize{Boxplots of the empirical
distribution of distances of the lower 95\% confidence bounds to the
true $AUC$ values (y-axis) for the bootstrap, and Learn2Evaluate without (L2E - BC) and with bias correction (L2E + BC) in the $N=100,$ $\nu=1000$ setting. The learning curve is fitted by an inverse power law and we determine
$n_{opt}$ by MSE minimization. Outliers are not depicted.}}
\label{fig:Boxplots confidence bounds1}
\end{figure}
\begin{figure}[H]
\begin{center}
\includegraphics[bb=0bp 50bp 550bp 305bp,clip,width=6in]{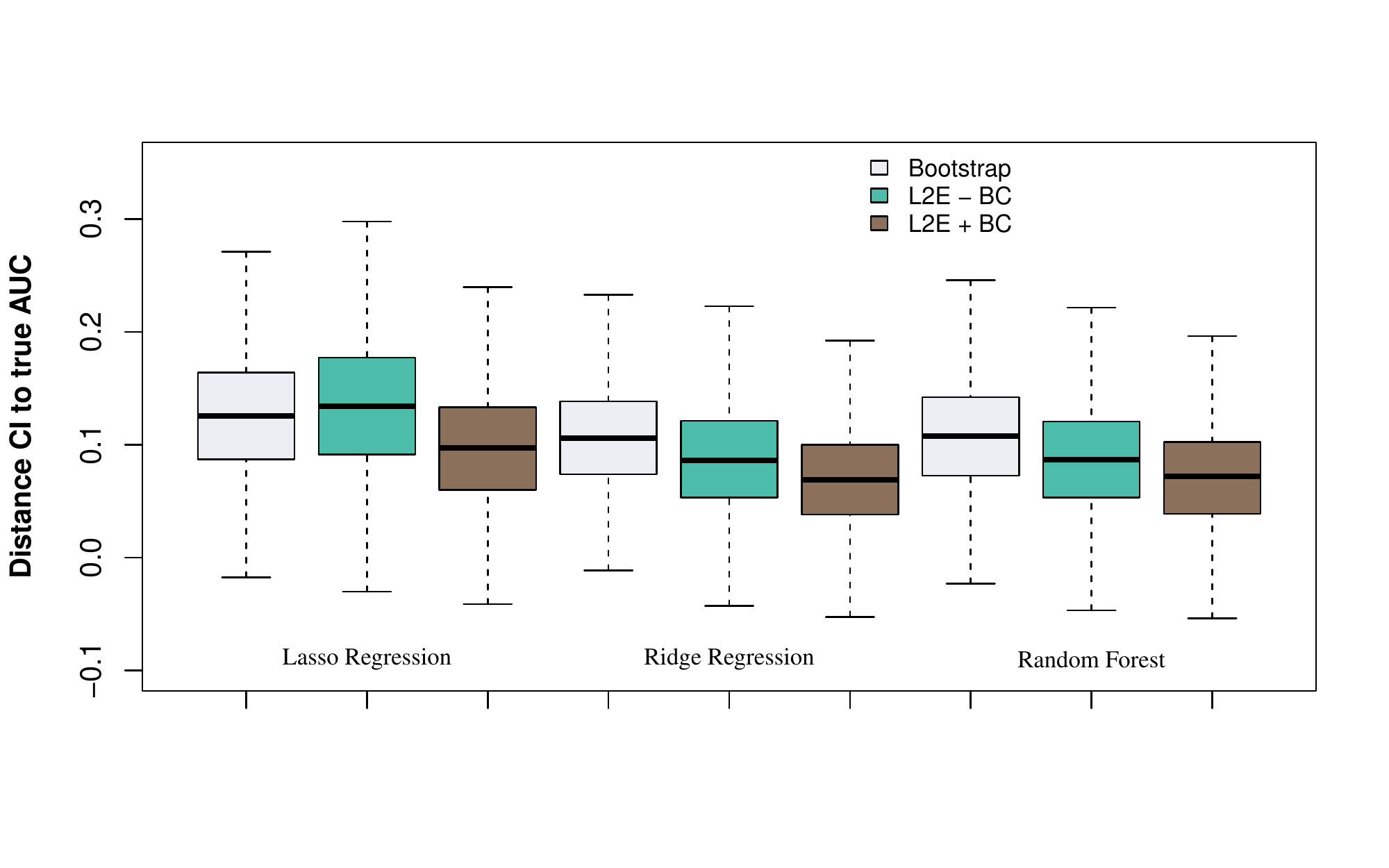}
\end{center}
\caption{\footnotesize{Boxplots of the empirical
distribution of distances of the lower 95\% confidence bounds to the
true $AUC$ values (y-axis) for the bootstrap, and Learn2Evaluate without (L2E - BC) and with bias correction (L2E + BC) in the $N=100,$ $\nu=100$ setting. The learning curve is fitted by an inverse power law and we determine
$n_{opt}$ by MSE minimization. Outliers are not depicted.}}
\label{fig:Boxplots confidence bounds2}
\end{figure}
\begin{figure}[H]
\begin{center}
\includegraphics[bb=0bp 50bp 550bp 305bp,clip,width=6in]{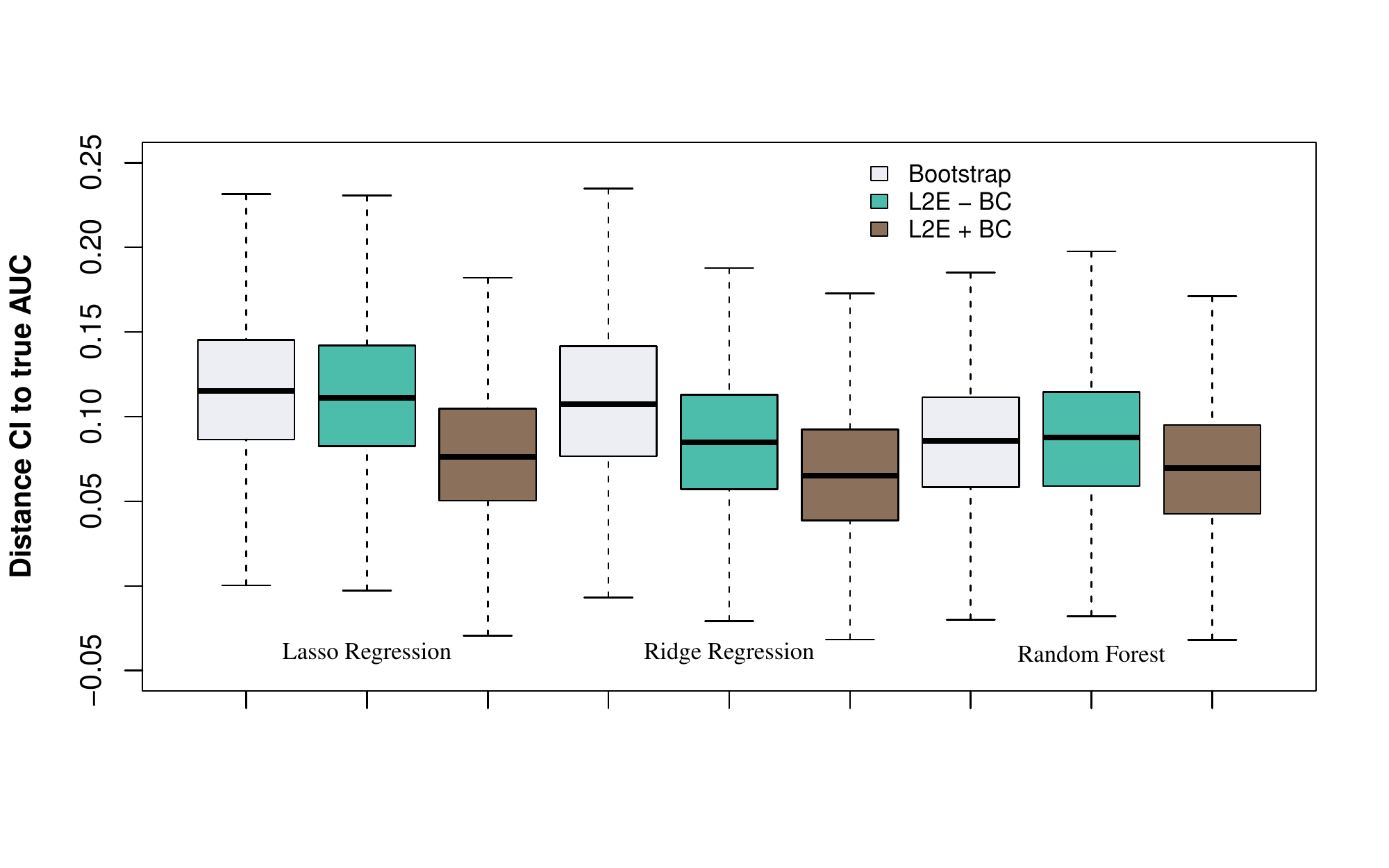}
\end{center}
\caption{\footnotesize{Boxplots of the empirical
distribution of distances of the lower 95\% confidence bounds to the
true $AUC$ values (y-axis) for the bootstrap, and Learn2Evaluate without (L2E - BC) and with bias correction (L2E + BC) in the $N=200,$ $\nu=1000$ setting. The learning curve is fitted by an inverse power law and we determine
$n_{opt}$ by MSE minimization. Outliers are not depicted.}}
\label{fig:Boxplots confidence bounds3}
\end{figure}
\begin{figure}[H]
\begin{center}
\includegraphics[bb=0bp 50bp 550bp 305bp,clip,width=6in]{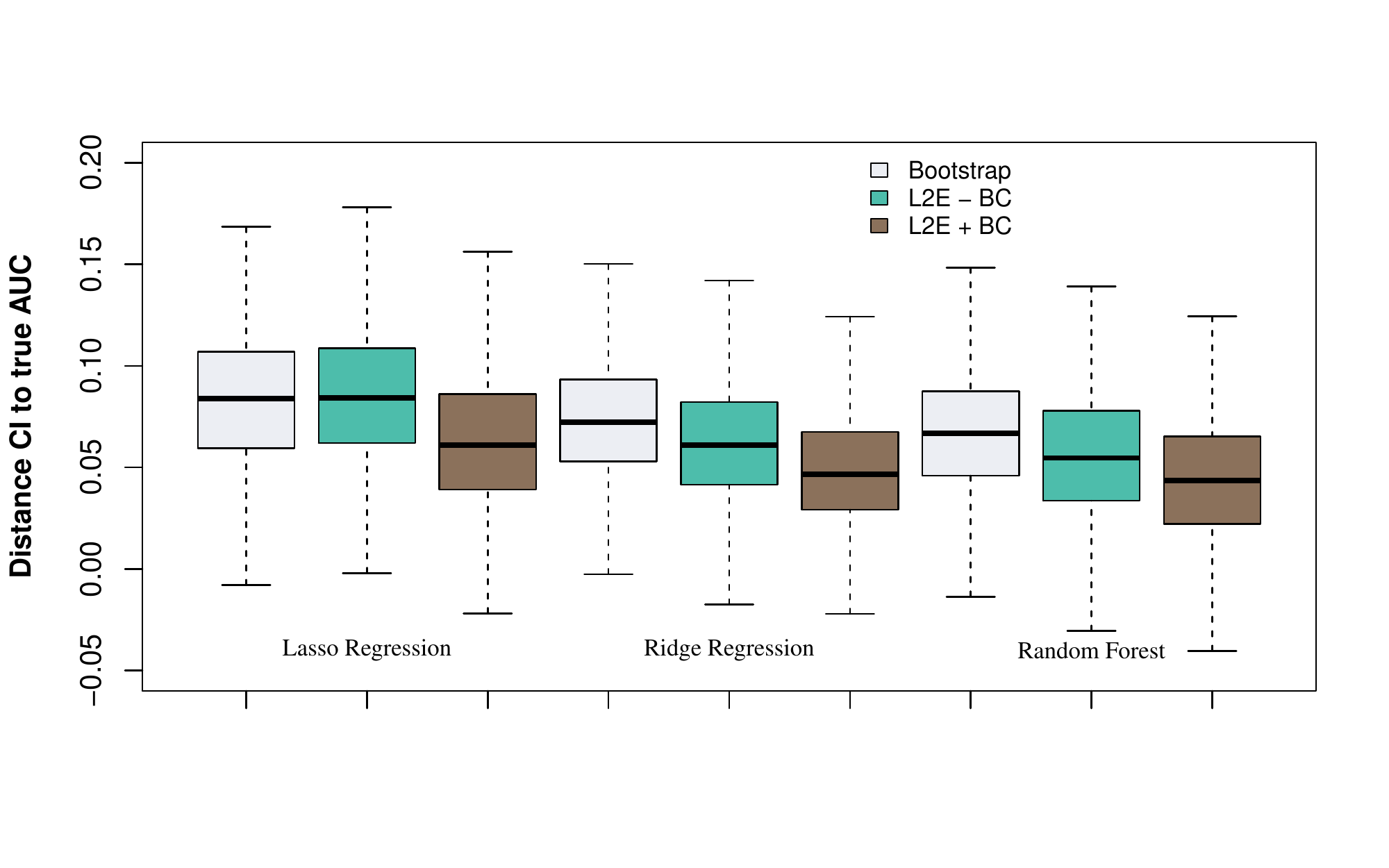}
\end{center}
\caption{\footnotesize{Boxplots of the empirical
distribution of distances of the lower 95\% confidence bounds to the
true $AUC$ values (y-axis) for the bootstrap, and Learn2Evaluate without (L2E - BC) and with bias correction (L2E + BC) in the $N=200,$ $\nu=100$ setting. The learning curve is fitted by an inverse power law and we determine
$n_{opt}$ by MSE minimization. Outliers are not depicted.}}
\label{fig:Boxplots confidence bounds4}
\end{figure}

\section{Three other applications of Learn2Evaluate}
\label{sec:application}

Here, we depict some additional learning curves in a classification setting. Learn2Evaluate is again applied to mRNA-Seq data extracted from blood platelets.
The same methodology as for Figure 2 in the main document is used.
We consider the following binary classifications: non-small-cell lung
cancer $(\mbox{NSCLC,}\quad n=60)$ versus pancreas cancer $(\mbox{Pancreas,}\quad n=35)$ 
using $p=18,582$ transcripts  (\autoref{fig:NSCLC vs Pancreas}), breast cancer $(\mbox{Breast,}\quad n=40)$ 
versus colorectal cancer $(\mbox{CRC,}\quad n=41)$ using
$p=18,410$ transcripts  (\autoref{fig:Breast vs CRC}), and colorectal cancer $(n=41)$ versus
pancreas cancer $(n=35)$ using $p=18,090$ transcripts  (\autoref{fig:CRC vs Pancreas}).
\begin{figure}[H]
    \centering
    
    \begin{subfigure}[t]{0.49\textwidth}
        \centering
        \includegraphics[bb=0bp 10bp 504bp 412bp,width=1\textwidth]{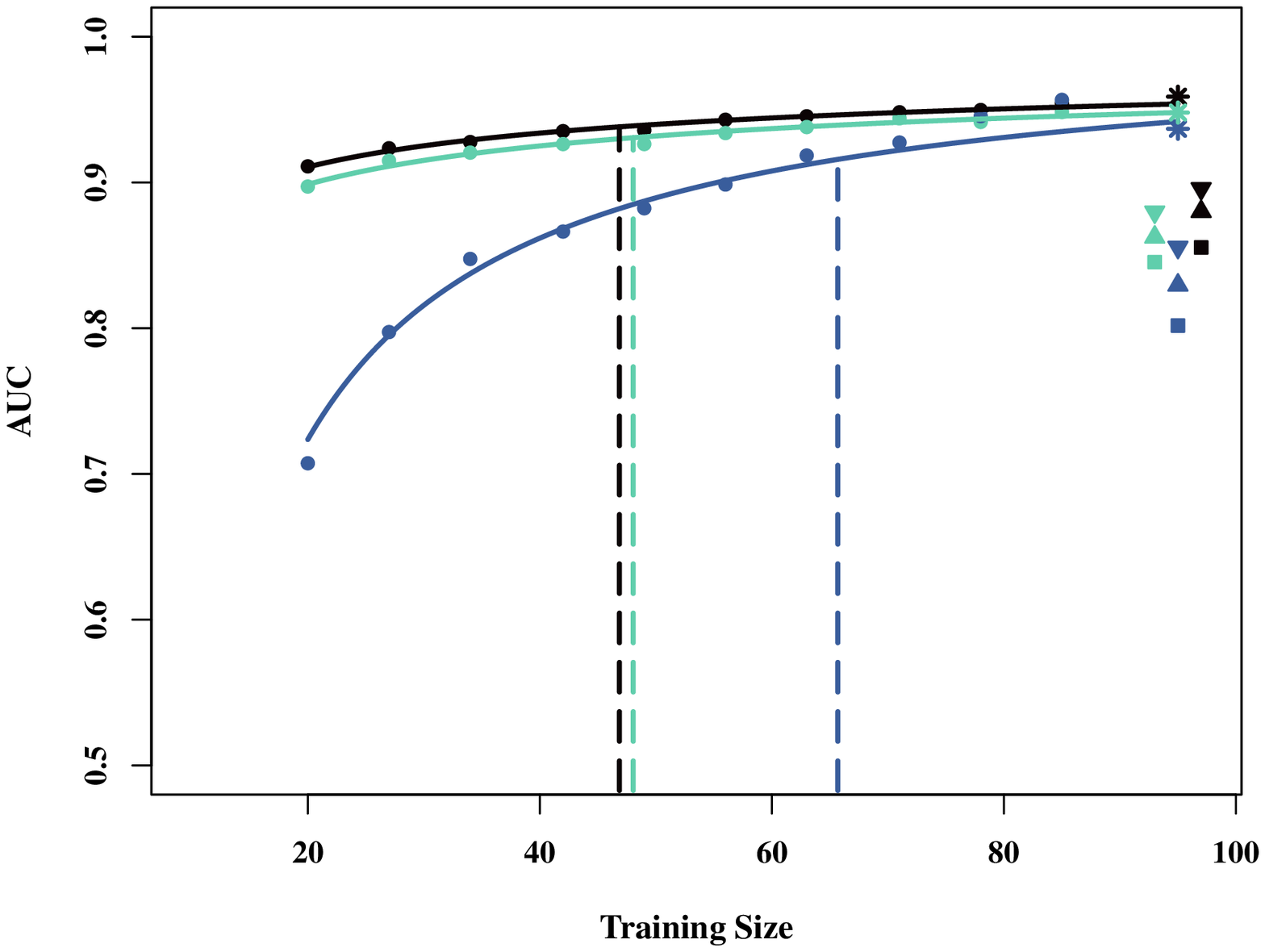}
        \caption{\footnotesize{NSCLC vs. Pancreas}}
        \label{fig:NSCLC vs Pancreas}
    \end{subfigure}
    \hfill
    \begin{subfigure}[t]{0.49\textwidth}
        \centering
	\includegraphics[bb=0bp 10bp 504bp 412bp,width=1\textwidth]{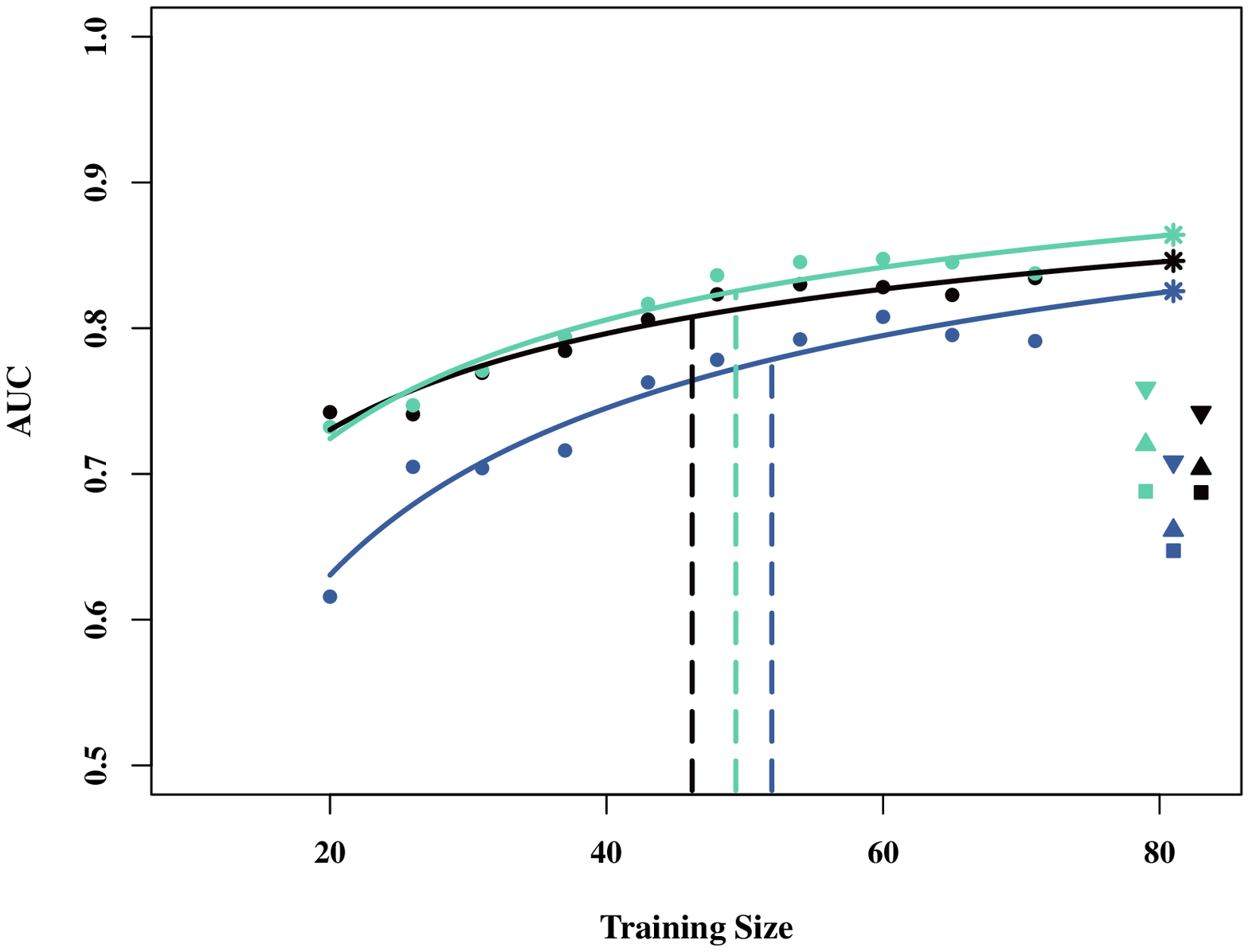}    
	\caption{\footnotesize{Breast vs. CRC}}
        \label{fig:Breast vs CRC}
    \end{subfigure}\\
     \begin{subfigure}[t]{0.49\textwidth}
        \centering
	\includegraphics[bb=0bp 10bp 504bp 412bp,width=1\textwidth]{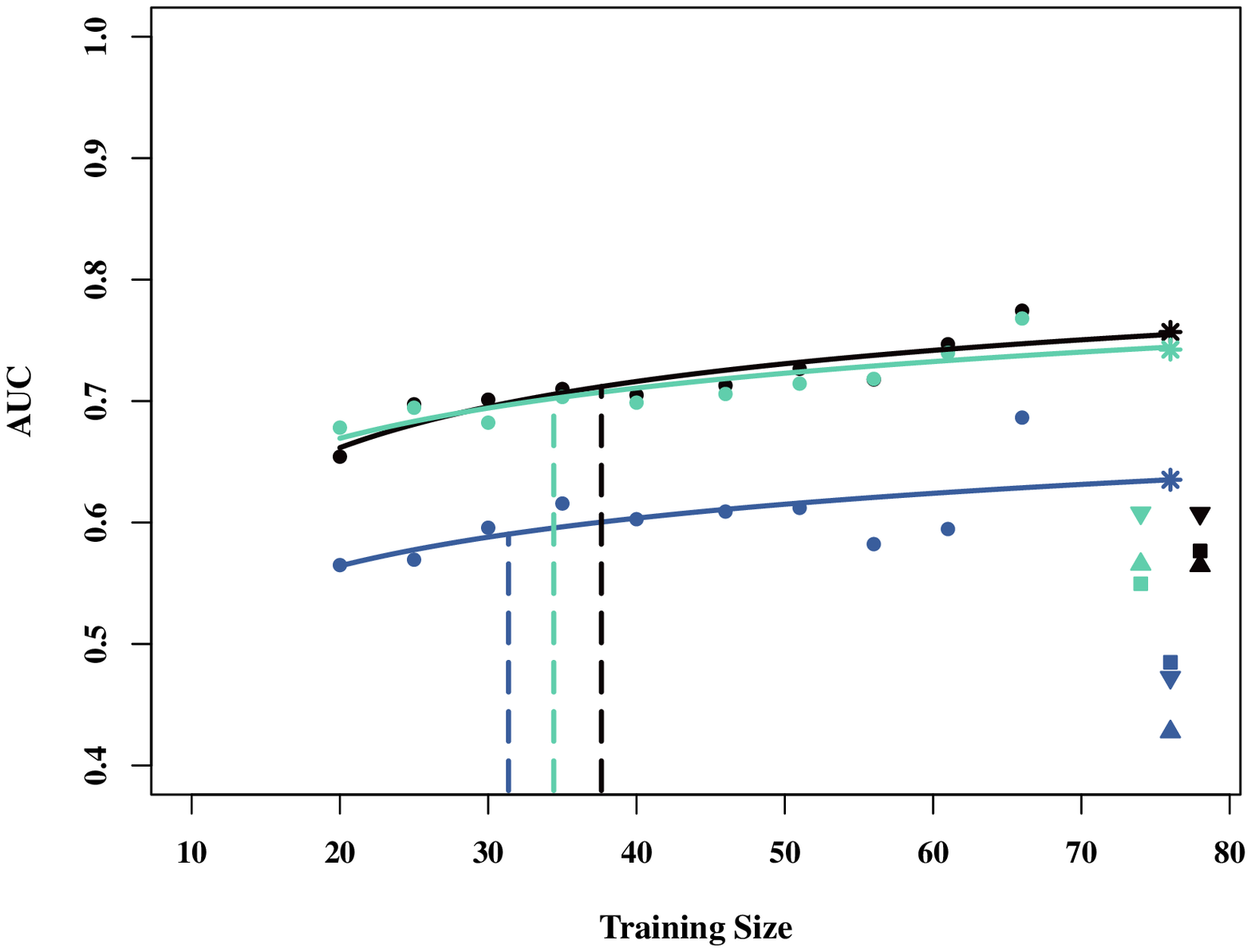}    
	\caption{\footnotesize{CRC vs. Pancreas}}
        \label{fig:CRC vs Pancreas}
     \end{subfigure}
     \hfill
     \begin{subfigure}[t]{0.49\textwidth}
        \centering
	\captionsetup[subfigure]{labelformat=empty}
	\includegraphics[bb=0bp 10bp 504bp 412bp,width=1\textwidth]{Figures/legend}    
	\label{fig:legend}
    \end{subfigure}
     \caption{\footnotesize{Learning curves for three classification
	tasks. In each task, we produce a learning curve for a ridge (black),
	a lasso (blue), and a random forest (green) model. For each learner, we depict the learning trajectory (dots)
	together with the inverse power law fit (solid line).
	We show confidence bounds of Learn2Evaluate with  ($L^{bc}_{n_{opt}}$,
	downward triangle) and without ($L_{n_{opt}}$, upward triangle) bias correction, and bootstrapping (Bootstrap, square). We use MSE minimization to 
       	determine $n_{opt}$ (dashed vertical line).}}
    \label{fig:Learning-curves}
\end{figure}
\section{Regression application}
\label{sec:Regression application}
To apply Learn2Evaluate in a regression setting, we use a DNA methylation data set described by \cite{akalin2020computational}.
Experiments have demonstrated that the methylation level of CpG islands are predictive of age \citep{DNAmeth,horvath2013dna}.
The processed data set is available via \url{https://github.com/JeroenGoedhart/Learn2Evaluate}. Data preprocessing is described in \autoref{sec:sofware}

We generate learning curves for ridge regression (black), lasso
regression (blue), and random forest (green) (\autoref{fig:Learning-curves-regression}). Learning curves are
obtained by the same procedure as in the simulation study and the application to mRNA-Seq data. For ten
homogeneously spaced subsample sizes, we estimate the predictive performance,
quantified by the PMSE, by repeated ($50$ times) hold-out estimation.
This renders the learning trajectory (dots, eq. 2 of the main document).
The maximum allowed subsample size equals $N-10.$ We then fit an inverse
power law (solid line, eq. 4 of the main document), which is decreasing and convex. 

We construct upper confidence bounds based on Learn2Evaluate without bias correction ($L_{n_{opt}}$, upward triangle), Learn2Evaluate with bias correction ($L^{bc}_{n_{opt}}$, downward triangle), and leave-one-out bootstrapping (Bootstrap, square).
We use MSE minimization (eq. 11 of the main document) to determine
an optimal training subsample size $n_{opt}$ (dashed vertical line). We use
\eqref{eq:PMSEVariance} to estimate the variance of the PMSE and to construct an $95\%$ upper confidence bound. 
We also display the point estimates of Learn2Evaluate (star).
\begin{figure}[H]
    \centering
     \begin{subfigure}[t]{0.49\textwidth}
        \centering
	\includegraphics[bb=0bp 10bp 504bp 412bp,width=1\textwidth]{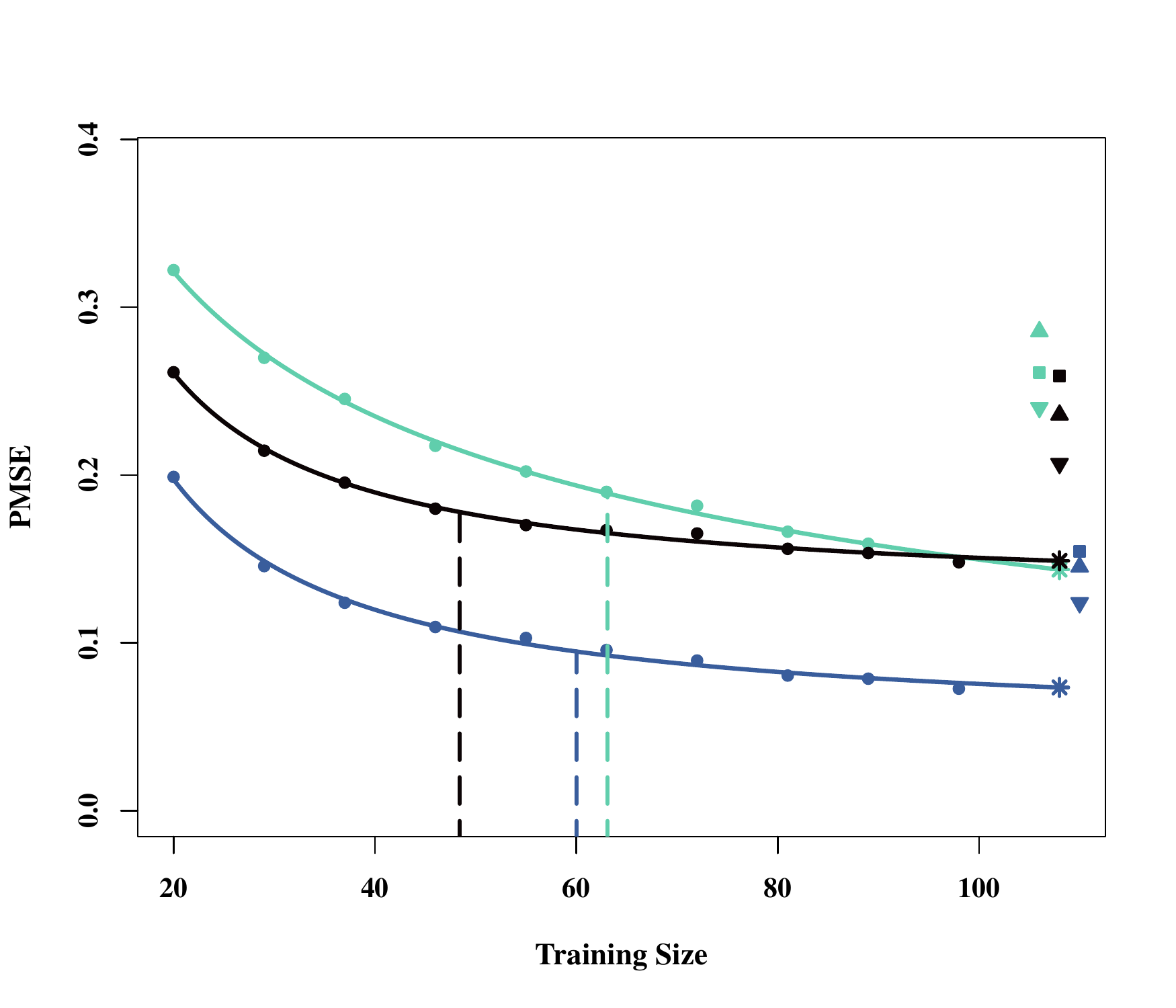}    
	\captionsetup[subfigure]{labelformat=empty}
     \end{subfigure}
     \hfill
     \begin{subfigure}[t]{0.49\textwidth}
        \centering
	\captionsetup[subfigure]{labelformat=empty}
	\includegraphics[bb=0bp 10bp 504bp 412bp,width=1\textwidth]{Figures/legend}    
	\label{fig:legend}
    \end{subfigure}
     \caption{\footnotesize{Learning curves for predicting age from DNA methylation data. 
        We produce a learning curve for a ridge (black/dark),
	a lasso (blue/medium), and a random forest (green/light) model. For
	the learning curve, we depict the learning trajectory (dots)
	together with the inverse power law fit (solid lines).
	For each learner, we show confidence bounds constructed by Learn2Evaluate without bias correction 
	($L_{n_{opt}}$, upward triangle), Learn2Evaluate with bias correction ($L^{bc}_{n_{opt}}$,
	downward triangle), and bootstrapping (Bootstrap, square). We used MSE minimization to 
	determine $n_{opt}$ (dashed vertical line).}}
    \label{fig:Learning-curves-regression}
\end{figure}

\autoref{fig:Learning-curves-regression} illustrates that the decreasing and convex power law fits the learning trajectories for the PMSE well.
The upper confidence bounds of Learn2Evaluate are in most cases tighter than the bootstrapped confidence bounds. Only the theoretically justified 
confidence bound $L_{n_{opt}}$ of random forest is more conservative. The bias-corrected confidence 
bounds $L^{bc}_{n_{opt}}$ are always tighter than the bootstrapped confidence bounds. Again, the optimal 
training sizes $n_{opt}$ nicely adapt to the saturation level of the curves. 

Lasso regression has a smaller PMSE than ridge regression and random forest in this application. 
Interestingly, the bias-corrected confidence bound ($L^{bc}_{n_{opt}}$) of lasso regression is lower-valued
than the point estimates of ridge regression and random forest, which suggests a clear difference between the learners.

\section{Varying the learning trajectory}
\label{sec:Varying the learning trajectory}

Here, we show the stability of the AUC point estimates of Learn2Evaluate, obtained in the application to blood-platelet mRNA-Seq data. 
In \autoref{tab:VaryingLearningTrajectory}, we depict these point estimates when a learning trajectory (eq. 2 in the main document) of $10$ subsample sizes 
is considered and when a learning trajectory of $20$ subsample sizes is considered. In both scenarios, the number of repeats for each subsample size is fixed at $50.$

\begin{table}[H]
\centering{}\caption{\footnotesize{Stability of AUC point estimates
for a varying length of the learning trajectory ($10$ and $20$ subsample
sizes).}}
\begin{tabular}{cccc}
 & \textbf{Ridge} & \textbf{Lasso} & \textbf{Random Forest}\tabularnewline
\hline 
\hline 
 & \multicolumn{3}{c}{\textit{Breast vs. NSCLC}}\tabularnewline
\hline 
$\boldsymbol{10}$ \textbf{subsample sizes} & 0.826 & 0.840 & 0.829\tabularnewline
$\boldsymbol{20}$ \textbf{subsample sizes} & 0.823 & 0.835 & 0.834\tabularnewline
\hline 
 & \multicolumn{3}{c}{\textit{Breast vs. CRC}}\tabularnewline
\hline 
\textbf{$\boldsymbol{10}$} \textbf{subsample sizes} & 0.843 & 0.828 & 0.847\tabularnewline
$\boldsymbol{20}$ \textbf{subsample sizes} & 0.849 & 0.833 & 0.850\tabularnewline
\hline 
 & \multicolumn{3}{c}{\textit{Breast vs. Pancreas}}\tabularnewline
\hline 
$\boldsymbol{10}$ \textbf{subsample sizes} & 0.888 & 0.918 & 0.856\tabularnewline
$\boldsymbol{20}$ \textbf{subsample sizes} & 0.890 & 0.922 & 0.853\tabularnewline
\hline 
 & \multicolumn{3}{c}{\textit{NSCLC vs. Control}}\tabularnewline
\hline 
$\boldsymbol{10}$ \textbf{subsample sizes} & 0.975 & 0.962 & 0.967\tabularnewline
$\boldsymbol{20}$ \textbf{subsample sizes} & 0.974 & 0.963 & 0.967\tabularnewline
\hline 
 & \multicolumn{3}{c}{\textit{NSCLC vs. Pancreas}}\tabularnewline
\hline 
\textbf{$\boldsymbol{10}$ subsample sizes} & 0.954 & 0.948 & 0.944\tabularnewline
\textbf{$\boldsymbol{20}$ subsample sizes} & 0.954 & 0.949 & 0.943\tabularnewline
\hline 
 & \multicolumn{3}{c}{\textit{CRC vs. Pancreas}}\tabularnewline
\hline 
\textbf{$\boldsymbol{10}$ subsample sizes} & 0.749 & 0.618 & 0.726\tabularnewline
\textbf{$\boldsymbol{20}$ subsample sizes} & 0.750 & 0.622 & 0.723\tabularnewline
\end{tabular}
\label{tab:VaryingLearningTrajectory}
\end{table}

\autoref{tab:VaryingLearningTrajectory} shows that increasing the size of the learning trajectory to $20$ subsample sizes 
does hardly change the performance estimates. The largest difference ($0.6\%$) is observed for lasso regression in the CRC vs. Pancreas classification task.  

\section{Details on data and software}
\label{sec:sofware}

\subsection{Details on mRNA-Seq data set}
RNA profiles, extracted from bloodplatelets, were obtained from 285 individuals, as described in \cite{best2015rna}.
Of those individuals, 55 belong to the control group, and the remaining individuals
have one of the in total six tumor types: breast cancer ($n=40$), colocteral cancer ($n=41$),
glioblastoma ($n=40$), hepatobiliary cancer ($n=14$), non-small-cell lung carciona ($n=60$),
and pancreatic cancer ($n=35$). Raw sequencing data are available in the GEO database 
(\href{https://www.ncbi.nlm.nih.gov/geo/query/acc.cgi?acc=GSE68086}{GEO: GSE68086})
and in the ArrayExpress repository (\href{https://www.ebi.ac.uk/arrayexpress/experiments/E-GEOD-68086/}{ArrayExpress ID: E-GEOD-68086}).
The raw data set has measurements of 57,736 transcripts. 

Data preprocessing is described by \cite{novianti2017better} in the supplementary material. We summarize their work here. 
Raw count data were first normalized by weighted trimmed mean of M-values \citep{robinson2010scaling}. Next, transcripts having a count less than 5 across samples were filtered out, rendering 19,629 transcripts. Finally, the data were transformed to a quasi-gaussian scale, 
i.e. $\sqrt{x_{ij}+\frac{3}{8}}-\sqrt{\frac{3}{8}}$, where $x_{ij}$ denotes
the count of the $j$th transcript in the $i$th sample. The processed data set is available via 
\url{https://github.com/JeroenGoedhart/Learn2Evaluate}. Prior to building the classifiers, we excluded transcripts with zero variance and we standardized the data.

\subsection{Details on DNA methylation data set}
The DNA methylation data set is described in \citep{akalin2020computational} and 
can be found in the R package \href{https://github.com/compgenomr/compGenomRData}{compGenomRData}.
Methylation levels of $p=27579$ CpG islands were measured for $108$ individuals.
Covariates having a standard deviation below $0.1$ were removed, which renders a total 
of $p=2290$ covariates. Next, the data was standardized, rendering the processed data set, which is available via \url{https://github.com/JeroenGoedhart/Learn2Evaluate}.

\subsection{Details on software}
R code to reproduce results is avaiable via \url{https://github.com/JeroenGoedhart/Learn2Evaluate}.
We use R version 4.0.3. The repository \href{https://github.com/JeroenGoedhart/Learn2Evaluate}{Learn2Evaluate} 
contains the data map, where the processed RNA sequencing and DNA methylation data  are located, the simulation-results map, which contains R scripts to reproduce
the simulation results in the paper (including seeds for the pseudo-random-number generator), and the R scripts map. The latter map includes the R script ``Learn2Evaluate", 
which contains the function to apply Learn2Evaluate to a data set of choice. Details are in the R script. 
The R scripts map also contains the script ``run\_Learn2\_Evaluate\_real\_data", which displays the code to reproduce the 
results in the application section.

The software depends on the following R packages: 
\href{https://cran.r-project.org/web/packages/mvtnorm/index.html}{mvtnorm (version 1.1-3)}, 
\href{https://cran.r-project.org/web/packages/glmnet/index.html}{glmnet (version 4.1-3)}, 
\href{https://cran.r-project.org/web/packages/cobs/index.html}{cobs (version 1.3-4)}, 
\href{https://cran.r-project.org/web/packages/randomForestSRC/index.html}{randomForestSRC (version 3.0.2)}, 
\href{https://cran.r-project.org/web/packages/pROC/index.html}{pROC (version 1.18.1)}, and 
\href{https://cran.r-project.org/web/packages/cvAUC/index.html}{cvAUC (version 1.1.0)}.